\documentclass[11pt]{article}
\newif\ifanonym
\anonymtrue
\anonymfalse

\usepackage[letterpaper, margin=1in]{geometry}
\usepackage{hyperref}
\hypersetup{colorlinks=true,allcolors=blue}
\usepackage{amsthm}
\usepackage{amsmath}
\usepackage{amsfonts}
\usepackage{import} 
\usepackage{xifthen}
\usepackage{mathtools}
\usepackage{arydshln}%
\usepackage{tikzit}

\tikzstyle{Title}=[font={\bfseries \small}]
\tikzstyle{new style 0}=[font={\tiny}]

\tikzstyle{Thick Line}=[-, thick]
\tikzstyle{Thick n}=[-, thick, fill={rgb,255: red,0; green,163; blue,0}]
\tikzstyle{Thick n^2}=[-, thick, fill={rgb,255: red,0; green,208; blue,0}, draw=none]
\tikzstyle{Thick n^3}=[-, thick, fill={rgb,255: red,125; green,255; blue,125}]
\tikzstyle{Thick Dashed}=[-, thick, dash dot, fill=none]

\usetikzlibrary{matrix,decorations.pathreplacing, calc, positioning}
\usepackage{xspace}
\usepackage{dsfont}
\usepackage{cleveref} 
\usepackage[cal=esstix]{mathalfa}
\usepackage{algorithm}
\usepackage[noend]{algpseudocode}
\newcommand{\mintcut}{\mathrm{cut}}

\newcommand{\N}{\mathbb{N}}
\newcommand{\g}{\mathcal{g}}

\newcommand{\R}{\mathbb{R}}
\newcommand{\Exp}[2]{\mathbb{E}_{ #2}\left[ #1 \right]}

\newcommand{\Probability}[1]{\Pr\left[ #1 \right]}

\newcommand{\T}{\mathcal{T}}

\newcommand{\minhecut}[2]{g_{#1}^{#2}}
\newcommand{\tOmega}{\tilde{\Omega}}

\newcommand{\splitaon}{g_e^{\text{aon}}}
\newcommand{\splitdaon}{g_e^{\text{d-aon}}}
\newcommand{\splitsml}{g_e^{\text{sml}}}
\newcommand{\splitprd}{g_e^{\text{prd}}}
\newcommand{\tildeoeps}[1]{\tilde{O}_\epsilon( #1 )}

\newcommand{\tO}{\tilde{O}}
\DeclareMathOperator{\poly}{poly}

\newcommand{\indic}[1]{ \mathds{1}_{\{#1\}} } %
\providecommand{\set}[1]{{\{#1\}}}
\newcommand{\eqdef}{\coloneqq}
\newcommand{\Benczur}{Bencz{\'{u}}r\xspace}

\newtheorem{theorem}{Theorem}[section]
\newtheorem{question}[theorem]{Question}
\newtheorem{claim}[theorem]{Claim}

\newtheorem{lemma}[theorem]{Lemma}
\newtheorem{corollary}[theorem]{Corollary}
\newtheorem{conjecture}[theorem]{Conjecture}
\newtheorem{definition}[theorem]{Definition}
\newtheorem{observation}[theorem]{Observation}

\title{Cut Sparsification and Succinct Representation \\ of Submodular Hypergraphs%
  \ifanonym
  \else 
    \thanks{This research was partially supported 
      by the Israel Science Foundation grant \#1336/23, 
      by a Weizmann-UK Making Connections Grant,
      by a Minerva Foundation grant,
      by the Israeli Council for Higher Education (CHE) via the Weizmann Data Science Research Center,
      and by a research grant from the Estate of Harry Schutzman.
    }
  \fi
}

\ifanonym
\author{Anonymous Authors}
\else
\author{
  Yotam Kenneth 
  \qquad
  Robert Krauthgamer
  \\ Weizmann Institute of Science
  \\ \texttt{\{yotam.kenneth,robert.krauthgamer\}@weizmann.ac.il}
}
\fi

\begin{document}

\maketitle

\begin{abstract}
In cut sparsification, all cuts of a hypergraph $H=(V,E,w)$ are approximated
within $1\pm\epsilon$ factor by a small hypergraph $H'$.
This widely applied method was generalized recently 
to a setting where the cost of cutting each hyperedge $e$ is provided by a splitting function $g_e: 2^e\to\R_+$.
This generalization is called a submodular hypergraph
when the functions $\{g_e\}_{e\in E}$ are submodular,
and it arises in machine learning, combinatorial optimization, and algorithmic game theory.

Previous work studied the setting where $H'$ is a reweighted sub-hypergraph of $H$,
and measured the size of $H'$  by the number of hyperedges in it.
In this setting, we present two results:
(i) all submodular hypergraphs admit sparsifiers of size polynomial in $n=|V|$ and $\epsilon^{-1}$; 
(ii) we propose a new parameter, called spread, and use it to obtain smaller sparsifiers in some cases.

We also show that for a natural family of splitting functions, relaxing the requirement that $H'$ be a reweighted sub-hypergraph of $H$ yields a substantially smaller encoding of the cuts of $H$ (almost a factor $n$ in the number of bits).
This is in contrast to graphs, where
the most succinct representation is attained by reweighted subgraphs.
A new tool in our construction of succinct representation is the notion of deformation,
where a splitting function $g_e$ is decomposed into a sum of functions of small description,
and we provide upper and lower bounds for deformation of common splitting functions.
\end{abstract}

\newpage

\section{Introduction}
\label{sec:introduction}
A powerful tool for many graph problems is sparsification,
where an input graph is replaced by a small graph
that preserves (perhaps approximately) certain properties, 
for example all the input graph's cuts~\cite{benczur1996approximating} or its spectrum~\cite{ST11,BSS14,JRT24}.
Downstream applications can then be executed on the small graph,
which improves the overall running time,
and the small graph can also be stored (or sent to another site)
instead of the input graph,
which improves the memory (or communication) requirements. 
The extensive research on cut sparsification
has started with the seminal work of \Benczur and Karger on cuts in graphs~\cite{benczur1996approximating},
and was later extended to hypergraphs~\cite{KK15, BST19, CKN20}
and to directed hypergraphs~\cite{SY19,CPCS21,KKTY21,OST22}.
In recent months the study of sparsification has been extended to even more general objects such as semi-norms \cite{JLLS23}, matroid quotients \cite{quanrud2022quotient}, and linear codes \cite{KPS24}.
We focus on sparsifying a generalized form of hypergraphs,
as explained next.

In recent years, the notion of cuts in a weighted hypergraph $H=(V,E,w)$
has been generalized
to a setting where each hyperedge $e\in E$
has a \emph{splitting function} $g_e:2^e\to \R_+$, such that $g_e(\emptyset)=0$,
and the \emph{value} of a cut $S\subseteq V$ is defined as
\begin{equation} \label{eq:cut}
  \mintcut_H(S) \eqdef \sum_{e\in E} g_e(S\cap e).
\end{equation}
Associating every $e\in E$ with the \emph{all-or-nothing} splitting function,
given by $\splitaon: S \mapsto w_e\cdot \indic{S\neq \emptyset,e}$,
clearly models an ordinary hypergraph $H=(V,E,w)$,
where the value of a cut is the total weight of hyperedges that intersect both sides;
in fact, a simple extension can model a directed hypergraph.
Such a generalized hypergraph $H=(V,E,\g)$,
where $\g=\left\{ g_e \right\}_{e\in E}$,
is called a \emph{submodular hypergraph} 
if all its splitting functions $g_e$ are submodular. 
Recall that a set function $g:2^e\to \R_+$ is \emph{submodular} if 
\[
  \forall S,T\subseteq e,
  \qquad
  g(S\cup T) + g(S\cap T) \leq g(S) + g(T). 
\]
Submodular hypergraphs are useful in clustering data with higher-order relations that are not captured by ordinary hyperedges~\cite{LM17,LM18,VBK20,LVSLG21,VBK21,ZLS22}.
For example, the \emph{small-side splitting} function,
given by $\splitsml: S \mapsto \min(|S|,|e\setminus S|)$, 
is employed when unbalanced cuts are preferable. 
Cut functions of submodular hypergraphs were studied also
under a different name of \emph{decomposable submodular functions}. 
A submodular function $f:2^V \to \R_+$ is called \emph{decomposable}
if it can be written as $f = \sum_i f_i$,
where each $f_i:2^V \to \R_+$ is submodular.
This notion is widely applied in data summarization~\cite{GK10,LB11, TIWB14},
where each $f_i$ is a submodular similarity function,
and the task of summarizing the data under a given budget $k$
is modeled by maximizing $f(S)$ over all $S\subset V$ of size $|S|\leq k$.
Decomposable submodular functions arise also in welfare maximization,
where each agent has a submodular utility function,
for instance in approximation algorithms \cite{Feige09,FV06} and in truthful mechanisms~\cite{DobzinskiS06,AssadiS20}.

We study how to succinctly represent all the cuts of a submodular hypergraph $H$
up to $1\pm\epsilon$ factor.
We examine two complementary approaches:
(1) \emph{sparsification},
which reduces the number of hyperedges,
i.e., $H$ is represented using a sparse $H'$; and
(2) \emph{deformation},
which replaces large hyperedges or complicated splitting functions
by new ones of low space complexity,
i.e., $H$ is represented using $H'$ whose hyperedges can be stored succinctly. 
These approaches can yield (separately and/or together)
a sparsifier $H'$ that can be encoded using a small number of bits.
More generally, we may consider a general encoding
that \emph{need not} rely on a sparsifier $H'$,
e.g., an explicit list of all the $2^{|V|}$ cut values.

Let us introduce some basic notation to make the discussion more precise.
Throughout, let $n\eqdef |V|$;
we write $\tO(t)$ or $\tOmega(t)$ to suppress a polylogarithmic factor in $t$,
and $O_\alpha(t)$ or $\Omega_\alpha(t)$ to hide a factor that depends only on $\alpha$. 
\begin{definition}[Sparsifier]
  A \emph{cut sparsifier} of \emph{quality} $1+\epsilon$ for $H=(V,E,\g)$,
  or in short a \emph{$(1+\epsilon)$-sparsifier},
  is a submodular hypergraph $H'=(V,E',\g')$ such that 
  \begin{equation}
    \forall S \subseteq V,
    \qquad
    \mintcut_{H'}(S) \in (1\pm\epsilon)\cdot\mintcut_H(S) .
    \label{sparsifier-1-eps-condition}
  \end{equation}
  The \emph{size} of the sparsifier is $|E'|$.
  We call $H'$ a \emph{reweighted subgraph} of $H$
  if $E'\subseteq E$ and each function $g'_e$ for $e\in E'$ is a scaling of $g_e$ 
  (i.e., $g'_{e} \equiv s_e g_e$ for some $s_e>0$). 
\end{definition}
\begin{question}[Sparsification]
    \label{q:sparsification}
    Do all submodular hypergraphs 
    admit a reweighted-subgraph sparsifier with few hyperedges,
    say $\poly(\epsilon^{-1} n)$?
    And which families of splitting functions admit even smaller sparsifiers,
    like $\tO_\epsilon(n^2)$ or even $\tO_\epsilon(n)$?
\end{question}

The first question (about a polynomial bound) was previously answered for several families of splitting functions
(see \Cref{sec:SparsificationSubmodular} for a detailed account), but despite this significant progress,
the case of general submodular splitting was left open in~\cite{RY22},
where the bound on the sparsifier size depends on $\g$
and is exponential in $n$ in the worst case. 
We answer this first question in the affirmative, and also 
address the second question by showing families of splitting functions that admit even smaller sparsifiers.

We further ask about a more general notion, 
of encoding an approximation of all the cuts of $H$,
which can potentially be more succinct than a sparsifier. 

\begin{question}[Succinct Representation]
  \label{q:succinct}
  What is the smallest encoding (in bits of space) 
  that stores a submodular hypergraph $H$ 
  so as to report $(1+\epsilon)$-approximation to every cut value? 
  In particular, what is the smallest number of bits $s=s(\epsilon,n)$ 
  that suffices to store a sparsifier for $H$? 
\end{question}

For simplicity, we ask above only about the existence
of a sparsifier or an encoding, 
but we are of course interested also in fast algorithms to build them.
Fortunately, an algorithmic solution follows from the existential ones
because our proofs are constructive.
Furthermore, the running times are polynomial under the assumption that every $g_e$ takes integral values and $\max_{S\subseteq e} g_e(S) \le \poly(n)$.\footnote{The running times of \Cref{theorem:upper-bound-sparsification-all} and \Cref{theorem:additive-approximation} are polynomial in general. \Cref{theorem:upper-bound-sparsification-finite-spread} is polynomial under the stated assumption.}

\subsection{Sparsification: All Submodular Hypergraphs}
\label{sec:SparsificationSubmodular}

We start with addressing \Cref{q:sparsification}. 
Our first result (proved in Section~\ref{sec:polynomial-submodular-sparsifiers-proof})
provides the first polynomial (in $n$) bound for all submodular splitting functions;
the previous bound, due to \cite{RY22}, 
was $O_{\epsilon}(n^2 B_H)$, where $B_H\eqdef \max_{e\in E} |\mathcal{B}(g_e)|$
and $\mathcal{B}(g_e)$ is the set of extreme points in the polytope of $g_e$.%
\footnote{A recent manuscript~\cite{KZ23} claims that the proof in~\cite{RY22}
  has a flaw and holds only for monotone submodular hypergraphs.
}
In general, $B_H$ can be exponential in $n$,
for example small-side splitting $\splitsml$ has $|\mathcal{B}(\splitsml)|=2^{\Theta(|e|)}$.
\begin{theorem}
  \label{theorem:upper-bound-sparsification-all}
  Every submodular hypergraph admits a $(1+\epsilon)$-sparsifier of size $O( \epsilon^{-2} n^3 )$, which is in fact a reweighted sub-hypergraph.
\end{theorem}
This bound is within factor $O_{\epsilon}(n)$ of the $\Omega(n^2/\epsilon)$ lower bound known for cut sparsification of directed hypergraphs \cite{OST22}.
We also show that if all the splitting functions are monotone (i.e., $g_e(S)\le g_e(T)$ for all $S\subseteq T$),
then the sparsifier size can be improved to $O_\epsilon(n^2)$.
Monotone submodular functions arise in many applications, however no sparsification bound was previously known for this family.%
\footnote{The running time of \cite{RY22} was improved in \cite{KZ23}, where a sparsifier of size $O(\epsilon^{-2}n^2B)$ for monotone functions
  with low curvature is constructed in polynomial time.
}
The formal statement and its proof appear in \Cref{sec:polynomial-submodular-sparsifiers-proof}.

\paragraph{Related Work.}
Previous work on sparsification focused mostly on specific splitting functions.
The study of this problem began with sparsifiers for undirected graph cut;
the current size bound is $O(\epsilon^{-2}n)$ edges~\cite{BSS14},
which improves over~\cite{benczur1996approximating}
and is known to be tight~\cite{ACKQWZ16,CKST19}.
Furthermore, sparsifiers of size $\tO_\epsilon(n)$ are known
for all-or-nothing splitting $\splitaon$ \cite{CKN20} (see also \cite{quanrud2022quotient}) and
for \emph{product splitting},
given by $\splitprd:S\mapsto |S|\cdot |e\setminus S|$ \cite{SHS16}.
In contrast, for the splitting that models cuts in a directed hypergraph,
the best construction known has size $\tO_\epsilon(n^2)$ \cite{OST22}, 
which is near-tight with an $\Omega(n^2/\epsilon)$ lower bound \cite{OST22};
this function, called \emph{directed all-or-nothing splitting}, is given by
$\splitdaon: S\mapsto \indic{e_T\cap S \ne \emptyset\ \wedge\ e_H \not\subseteq S}$,
where $e_H,e_T\subseteq e$ are the hyperedge's head and tail, respectively.
A recent result is more general and shows that
the entire family of symmetric splitting functions
admits sparsifiers of size $\tildeoeps{n}$ \cite{JLLS23}.

\Cref{fig:sparsifier-families} depicts several families of splitting functions 
and the sparsification bounds known for them, 
including our results from above and from \Cref{sec:SparsificationFiniteSpread}.

\paragraph{Techniques.}
Our sparsification method follows the importance-sampling approach,
which has been used extensively in the literature.
Every hyperedge $e\in E$ is assigned an importance $\sigma_e$,
and sampled with probability $p_e$ that is (at least) proportional to $\sigma_e$,
and the splitting function of every sampled $e$ is scaled by $1/p_e$.
The expected sparsifier size is clearly proportional to $\sum_{e\in E}\sigma_e$.

A standard method to set the importance of a hyperedge $e\in E$,
is to consider all its possible cuts,
namely, $\sigma_e \coloneqq \max_{S\subseteq V} g_e(S\cap e)/\mintcut_H(S)$,
and this method was indeed used in~\cite{RY22}. 
Bounding $\sum_{e\in E}\sigma_e$ naively by replacing the maximization over ${S\subseteq V}$ by summation yields an exponential size bound.
An improved bound was given in~\cite{RY22}
based on a quantity $B_H$ related to the polytopes of the splitting functions.
Unfortunately, this improved bound is still exponential for many families of splitting functions.

Our main contribution is to identify a set of "basic" quantities for each hyperedge $e$
that can serve as coarse approximations of its splitting function $g_e$.
These approximations allow us to define new sampling probabilities and achieve an improved size bound:
Given $e\in E$, define the minimum directed cut between $u,v\in V$
to be $g_{e}^{u\to v} \coloneqq \min_{S\subseteq V: u\in S, v\not \in S} g_e(S\cap e)$;%
\footnote{The most natural case is $u,v\in e$,
  but considering all $u,v\in V$ streamlines the presentation. 
}
then our main technical lemma bounds $g_e(\cdot)$ from below and from above by
\begin{equation}
  \label{eq:introduction-minimum-directed-cut}
  \forall S\subseteq V,
  \qquad
  \max_{u \in S,v \in V\setminus S} \minhecut{e}{u \to v}
  \le g_e(S\cap e)
  \le \sum_{u\in S,v \in V\setminus S} \minhecut{e}{u \to v}
  ;
\end{equation}
The lower bound holds by definition,
and the upper bound is analogous to bounding the value of a graph cut
by the sum of the maximum flows between all pairs of vertices across the cut.
It is well-known that importance sampling will produce a sparsifier
even if $\sigma_e$ is replaced with an over-estimate for it.
We replace $\sigma_e$ with 
$\rho_e \coloneqq \sum_{(u,v) \in V\times V} \minhecut{e}{u \to v}/\sum_{f\in E} \minhecut{f}{u\to v}$,
which we can easily see is an over-estimate, i.e., $\rho_e \ge \sigma_e$, 
by using the two bounds from~\eqref{eq:introduction-minimum-directed-cut}
to verify that
\begin{equation*}
  \forall S\subseteq V,
  \qquad
  \frac{g_e(S\cap e)}{\mintcut_H(S)}
  = \frac{g_e(S\cap e)}{\sum_{f\in E} g_f(S\cap f)}
  \le \sum_{u\in S, v\in V\setminus S} \frac{\minhecut{e}{u\to v}}{\sum_{f\in E} \minhecut{f}{u\to v}}
  \le \rho_e
  .
\end{equation*}
The expected number of hyperedges in the sparsifier $H'$
equals to $\sum_{e\in E} \rho_e$ times an amplification factor $M$,
where $M=O(\epsilon^{-2}n)$ is sufficient by standard arguments
(a concentration bound and a union bound). 
The crux here is that it is easy to bound $\sum_{e\in E} \rho_e \leq O(n^2)$,
basically swapping the order of a double summation.
Another advantage of $\rho_e$ is that it can be computed in polynomial time,
while computing $\sigma_e$ requires maximizing the ratio of two submodular functions,
which is NP-hard in general.

In the monotone case, we follow the same approach but employ a simpler over-estimate 
$\rho_e' \coloneqq \sum_{v\in e} g_e(\left\{ v \right\})/\mintcut_H(\left\{ v \right\})$.
The proof is similar to the general case,
except that instead of~\eqref{eq:introduction-minimum-directed-cut}
we use the straightforward bound 
\begin{equation*}
  \forall S\subseteq V,
  \quad
  \max_{v\in S} g_e(\{v\}\cap e)
  \le
  g_e(S\cap e)
  \le \sum_{v\in S}g_e(\{v\}\cap e)
  .
\end{equation*}
\begin{figure}
    \begin{tikzpicture}
	\begin{pgfonlayer}{nodelayer}
		\node [style=Title] (6) at (0.25, -1.5) {finite-spread submodular};
		\node [style=new style 0] (8) at (0.25, -2.25) {$|E'| = \min(\mathbf{\mu n}, \mathbf{n^3},  n^2 B_H)$};
		\node [style=Title] (17) at (-8, -3) {symmetric};
		\node [style=new style 0] (20) at (-8, -3.75) {$|E'|=n$};
		\node [style=Title] (26) at (-7.5, -6) {hypergraph cuts};
		\node [style=Title] (31) at (0.75, -6.75) {cardinality based};
		\node [style=new style 0] (33) at (0.75, -7.5) {$|E'|=\min(\mathbf{n^2},\mathbf{\mu n })$};
		\node [style=Title] (34) at (9.5, -3.75) {monotone};
		\node [style=new style 0] (36) at (9.5, -4.5) {$|E'| = \min(\mathbf{\mu n}, \mathbf{n^2})$};
		\node [style=Title] (44) at (-7.5, -7.25) {graph cuts};
		\node [style=none] (55) at (4, -6.25) {};
		\node [style=none] (56) at (15, -6.25) {};
		\node [style=none] (76) at (-13.75, -6.25) {};
		\node [style=none] (77) at (-2.5, -6.25) {};
		\node [style=none] (78) at (-11.75, -6.5) {};
		\node [style=none] (79) at (-3.25, -6.5) {};
		\node [style=none] (80) at (-10, -7.25) {};
		\node [style=none] (81) at (-4.75, -7.25) {};
		\node [style=none] (86) at (-12.5, -7) {};
		\node [style=none] (87) at (4.75, -7) {};
		\node [style=none] (89) at (-16, 2.5) {};
		\node [style=none] (90) at (16.5, 2.5) {};
		\node [style=none] (91) at (16.5, -12) {};
		\node [style=none] (92) at (-16, -12) {};
		\node [style=Title] (93) at (-11.5, 1.25) {submodular functions};
		\node [style=new style 0] (94) at (-11.5, 0.5) {$|E'| = \min(\mathbf{n^3},  n^2 B_H)$};
		\node [style=none] (96) at (-15.75, -7) {};
		\node [style=none] (97) at (16.25, -7) {};
		\node [style=Title] (98) at (9.5, -7.75) {matroid rank};
		\node [style=new style 0] (99) at (9.5, -8.5) {$|E'|=\min(\mathbf{n^2}, \mathbf{\mu n})$};
		\node [style=none] (100) at (5, -7.75) {};
		\node [style=none] (101) at (14, -7.75) {};
		\node [style=Title] (102) at (11.25, 0.5) {directed hypergraphs};
		\node [style=new style 0] (103) at (11.25, -0.25) {$|E'|=n^2$};
		\node [style=none] (104) at (6.75, 0.5) {};
		\node [style=none] (105) at (15.75, 0.5) {};
	\end{pgfonlayer}
	\begin{pgfonlayer}{edgelayer}
		\draw [style=Thick Line, in=90, out=90, looseness=1.25] (56.center) to (55.center);
		\draw [style=Thick Line, in=-90, out=-90, looseness=1.25] (55.center) to (56.center);
		\draw [style=Thick Line, in=90, out=90, looseness=1.25] (77.center) to (76.center);
		\draw [style=Thick Line, in=-90, out=-90, looseness=1.25] (76.center) to (77.center);
		\draw [style=Thick Line, bend left=270, looseness=0.50] (79.center) to (78.center);
		\draw [style=Thick Line, in=-90, out=-90, looseness=0.75] (78.center) to (79.center);
		\draw [style=Thick Line, in=450, out=90, looseness=0.50] (81.center) to (80.center);
		\draw [style=Thick Line, in=-90, out=-90, looseness=0.50] (80.center) to (81.center);
		\draw [style=Thick Line, bend left=270, looseness=0.48] (87.center) to (86.center);
		\draw [style=Thick Line, in=-90, out=-90, looseness=0.35] (86.center) to (87.center);
		\draw [style=Thick Line] (91.center)
			 to (90.center)
			 to (89.center)
			 to (92.center)
			 to cycle;
		\draw [style=Thick Line] (97.center)
			 to [in=-90, out=270, looseness=0.50] (96.center)
			 to [bend left=90, looseness=0.75] cycle;
		\draw [style=Thick Line, bend left=270, looseness=0.25] (101.center) to (100.center);
		\draw [style=Thick Line, in=-90, out=-90, looseness=0.50] (100.center) to (101.center);
		\draw [style=Thick Line, bend left=270, looseness=0.50] (105.center) to (104.center);
		\draw [style=Thick Line, in=-90, out=-90, looseness=0.50] (104.center) to (105.center);
	\end{pgfonlayer}
\end{tikzpicture}
    \caption{Sparsification bounds for various families of submodular functions,
      omitting for simplicity $\poly(\epsilon^{-1}\log n)$ factors.}
    \label{fig:sparsifier-families}
\end{figure}
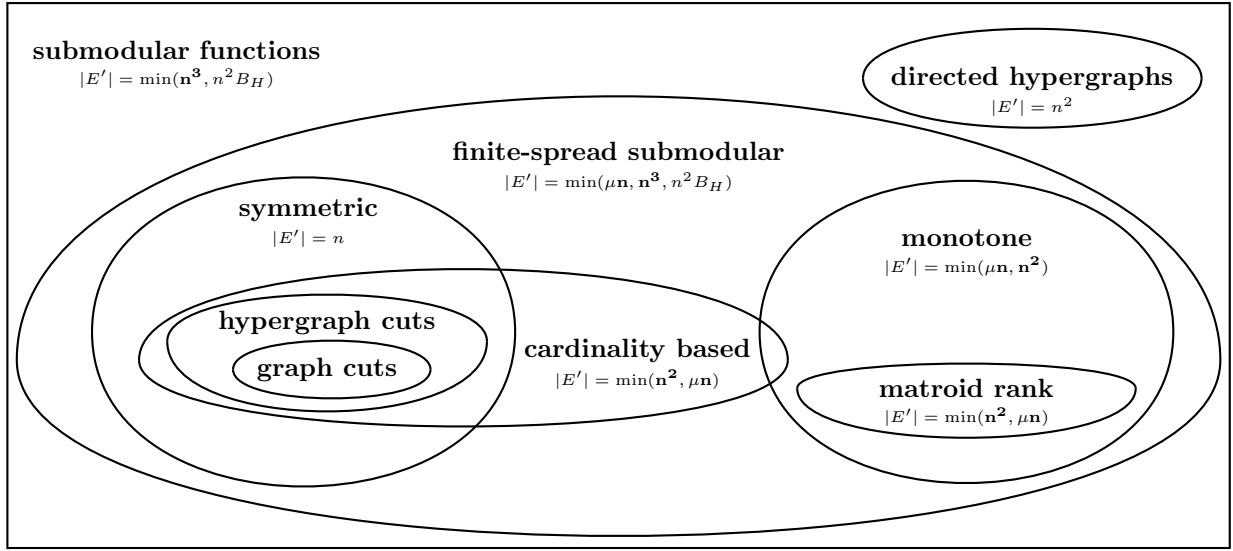

\subsection{Sparsification: Parameterized by Spread}
\label{sec:SparsificationFiniteSpread}

We already know that submodular splitting functions
can have very different optimal sparsification bounds,
see e.g.\ the bounds $\tilde{\Theta}_\epsilon(n)$ and $\tilde{\Theta}_\epsilon(n^2)$ mentioned above.
However, there are too many submodular functions to analyze each one separately,
and we thus seek a parameter that can control the sparsifier size. 
Our approach is inspired by the notion of imbalance in a directed graph $G=(V,E,w)$,
defined as the worst ratio between antiparallel edge weights,
i.e., $\beta_G \eqdef \max \{w(i,j)/w(j,i):\ i,j\in V\}$. 
This parameter can be used to show that every directed graph
admits a sparsifier of size $\tO_\epsilon(\beta_G n)$.%
\footnote{This condition can actually be relaxed significantly to
$\beta_G\eqdef \max \{\mintcut_G(S)/\mintcut_G(\bar S):\ S\subset V\}$,
and the same sparsification bound still holds \cite{CPCS21}.
}
For submodular hypergraphs, we propose an analogous parameter, 
which is basically the ratio between the maximum and minimum values of the splitting function, excluding certain trivial cuts. 

\begin{definition}[Spread]
  For hyperedge $e\in E$ with splitting function $g_e$,
  let $W_e\coloneqq\left\{ \emptyset \right\}$,
  unless $g_e(e)=0$ in which case $W_e\coloneqq \{\emptyset, e\}$.  
  The \emph{spread} of $e$ is 
  \begin{equation} \label{eq:spread}
    \mu_e \eqdef 
    \frac {\max_{T\subseteq e} g_e(T)} {\min_{ S\subseteq e: S\notin W_e} g_e(S)} .
  \end{equation}
\end{definition}

Our third result (proved in Section \ref{sec:finite-spread-sparsifiers-proof})
constructs a sparsifier 
whose size depends on the spread of the input $H$,
defined as $\mu_H \eqdef \max_{e\in E} \mu_e$.
By convention, the spread $\mu_e$ is called \emph{finite}
if it is well-defined (the denominator in \eqref{eq:spread} is non-zero),
and similarly $\mu_H$ is called finite
if it is well-defined (all the terms $\mu_e$ are finite).

\begin{theorem}[Sparsifier Parameterized by Spread]
  \label{theorem:upper-bound-sparsification-finite-spread}
  Every submodular hypergraph $H=(V,E,\g)$ with finite spread
  admits a $(1+\epsilon)$-sparsifier of size $\tilde{O}(\epsilon^{-2}\mu_H n)$, which is a sub reweighted-subgraph.
\end{theorem}

Many natural submodular functions have finite spread,
and in many common cases even $\mu_H \leq n$.
This can be seen, for example, in an easy application of \Cref{theorem:upper-bound-sparsification-finite-spread} to approximation of coverage functions, see \Cref{appendix:coverage} for details.
Another example is the sparsification of the capped version of small-side splitting, 
given by $g_e: S\mapsto \min(|S|, |e\setminus S|, K)$ for $K>0$,
which clearly has spread $\mu_e\leq K$.
This function is part of a much larger family, cardinality-based splitting functions, 
a notion formalized in \cite{VBK22} as follows: 
A submodular function $g_e:2^e\to \R_+$ is called \emph{cardinality-based}
if there exists a function $f_e:[|e|]\to \R_+$ such that $g_e: S\mapsto f_e(|S|)$.
Cardinality-based functions, which are commonly used in submodular hypergraph clustering,
all have spread $\mu_e\le n$, 
which is an easy consequence of the symmetry and subadditivity of $g_e$. 
By \Cref{theorem:upper-bound-sparsification-finite-spread},
these splitting function admit a $(1+\epsilon)$-sparsifier of size $\tilde{O}(\epsilon^{-2}n^2)$,
which is the first bound for this family.

It is easily verified that for monotone splitting functions,
the spread is approximately equal to the imbalance,
when we generalize the imbalance from above to hyperedges by
$\beta_e \eqdef \max\{g_e(S)/g_e(e\setminus S): S\subset V\}$.%
\footnote{For a monotone $g_e$,
  the spread is $\mu_e = g_e(V)/\min_{v\in V} g_e(\left\{ v \right\})$
  and the imbalance is $\beta_e = \max_{v\in  V} g_e(V\setminus\left\{ v \right\})/g_e(\left\{ v \right\})$,
  and they differ by at most a constant factor by the subadditivity of $g_e$. 
}
Hence, we immediately obtain the following.
\begin{corollary}
  Every finite-spread monotone splitting function admits a $(1+\epsilon)$-sparsisfier of size $\tilde{O}(\epsilon^{-2}\beta_H n)$.
\end{corollary}

Two other examples of commonly used monotone functions with finite spread
are set-coverage functions (defined in \Cref{appendix:coverage})
and the \emph{matroid-rank functions},%
\footnote{For a matroid with ground set $e$ and independent sets $\mathcal{I}$,
  the rank function is given by
  $g_e: S\mapsto \max_{T\subseteq S: T\in \mathcal{I}} |T|$.
  This rank function is submodular and monotone.
}
which have $\mu_e=r$ where $r$ is the rank of the matroid.

We remark that spread does not fully characterize the sparsifier size.
Indeed, symmetric functions can have a large spread $\mu_e$
but still admit $\tO_\epsilon(n)$ sparsifier due to \cite{JLLS23},
consider e.g. product splitting $\splitprd$ which has $\mu_e = O(n)$. 
Furthermore, directed all-or-nothing splitting $\splitdaon$
does not have finite spread,
and nevertheless admits a sparsifier of size $\tO_\epsilon(n^2)$ \cite{OST22}.
\Cref{fig:sparsifier-families} depicts different families of splitting functions
including that of finite spread, 
and the sparsification bounds known for them. 

\paragraph{Techniques.}
Our technique is based on approximate $H$ as an undirected hypergraph and use the sampling probabilities of \cite{CKN20} but amplified by $\mu_e$ for each hyperedge.
This is a known technique in generalizing sampling mechanisms.
Our main contribution is to identify the spread as a relevant and useful parameter.
We remark that the generalization of balance, which is known to control the size of sparsifier in directed graphs, to submodular hypergraphs does not suffice for sparsification.
Furthermore, we prove that the spread also characterizes other traits of splitting function, such as the deformation lower bound.

\subsection{Succinct Representation}
\label{sec:SuccintRepresentation}

We provide the first example of submodular splitting functions
for which sparsifiers that are not subgraphs
are provably (much) more succinct
than sparsifiers that are reweighted subgraphs.%
\footnote{Previously, a non-subgraph sparsifier 
  was shown in \cite{ADKKP16} for small-side splitting,
  however it optimizes the number of hyperedges and not the encoding size. 
}
To be more precise, we exhibit a natural family of splitting functions, 
where the former $(1+\epsilon)$-sparsifiers take only $\tO_\epsilon(n)$ bits
(\Cref{corollary:additive-small-representation}), 
while the latter $(1+\epsilon)$-sparsifiers require $\tOmega_\epsilon(n^2)$ bits
(\Cref{theorem:bit-lower-bound-additive-reweighted-subgraph}).
It follows that a reweighted subgraph \emph{need not be} the smallest encoding
that stores a $(1+\epsilon)$-approximation of the cuts values,
and by a wide margin! 

Our plan for constructing a succinct representation has two stages.
The first stage creates a $(1+\epsilon)$-sparsifier $H'$,
by deforming each $e\in E$ into multiple small hyperedges.
The second stage computes for this $H'$ a $(1+\epsilon)$-sparsifier $H''$
that is a reweighted subgraph.
It then follows that $H''$ is a $(1+\epsilon)^2$-sparsifier,
and has a few hyperedges that are all small.

\begin{definition}
  A splitting function $g_e:2^e\to\R_+$ on hyperedge $e$
  is called \emph{$(1+\epsilon)$-approximable with support size $p$}
  if there are submodular functions $g_{e_i}:2^{e_i}\to\R_+$ for $i=1,\ldots,r$,
  each on a hyperedge $e_i\subseteq e$ of size $|e_i|\le p$,
  such that
  \[
    \forall S\subseteq e,
    \qquad
    \sum_{i=1}^r g_{e_i}(S\cap e_i) \in (1\pm \epsilon)g_e(S)
    .
  \]
\end{definition}

Our example is the family of \emph{additive} splitting functions,
defined as functions $g_e$ that can be written
as either $g_e: S\mapsto \min(|S|,K)$
or $g_e:S\mapsto \min(|S|,|e\setminus S|,K)$ for some $K>0$.
The next theorem (proved in Section \ref{sec:deformation-additive-splitting})
achieves the first stage in our plan above;
it shows that additive functions can be $(1+\epsilon)$-approximated
by creating several copies of $e$ and sampling the vertices. 

\begin{theorem}[Deformation of Additive Functions]
\label{theorem:additive-approximation}
Let $g_e$ be an additive splitting function on hyperedge $e$. 
Then $g_e$ can be $(1+\epsilon)$-approximated with support size $O( \epsilon^{-2}(|e|/K)\log |e| )$.
\end{theorem}

Following our plan, suppose that given an input $H$,
we first apply \Cref{theorem:additive-approximation} 
to obtain a sparsifier $H'$ with small support size.
The construction of $H'$ also implies that it has small spread,
$\mu_{H'}  \le O(\epsilon^{-2} \log n)$.
Applying \Cref{theorem:upper-bound-sparsification-finite-spread}
on $H'$ we obtain a succinct representation $H''$.
A straightforward encoding of $H''$ then proves the following corollary (see \Cref{sec:deformation-additive-splitting}).

\begin{corollary}[Additive Functions admit Small Representation]
\label{corollary:additive-small-representation}
Let $H=(V,E,\set{g_e})$ be a submodular hypergraph such that every $g_e$ is additive with parameter $K_e>0$,
and let $\hat{K} \eqdef \min_{e\in E} K_e/|e|$ be a normalized bound on $K_e$ over all hyperedges.
Then $H$ admits a $(1+\epsilon)$-sparsifier 
with encoding size $O( \epsilon^{-6}\hat{K}^{-1} n\log^4 n )$ bits.
\end{corollary}

The next theorem (proved in \Cref{sec:encoding-size-lower-bounds}) 
shows that reweighted-subgraph sparsifiers 
of additive functions require $\Omega( n^2 )$ bits in the worst-case.
Putting this together with our succinct representation
from \Cref{corollary:additive-small-representation},
we conclude that relaxing the (natural) restriction to reweighted subgraphs
improves the space complexity by a factor of $\tOmega_\epsilon(n\hat{K})$, observe that this can be $\tOmega_\epsilon(n)$ when $\hat{K}\in \Omega(1)$.

\begin{theorem}[Reweighted Sparsifiers Require $\Omega( n^2)$ Bits]
  \label{theorem:bit-lower-bound-additive-reweighted-subgraph}
  There exists a family $\mathcal{H}$ of hypergraphs with additive splitting functions with parameter $1\le K \le n/3$,
  such that encoding a reweighted-subgraph $(1+\epsilon)$-sparsifier
  for an input $H\in \mathcal{H}$ requires $\Omega(n^2)$ bits. 
\end{theorem}

This lower bound is surprising because in the case of undirected graphs,
the best encoding size is achieved by a reweighted-subgraph sparsifier \cite{BSS14, ACKQWZ16,CKST19}.
Our proof is based on a technical lemma that can be applied to many cardinality-based splitting functions.
Furthermore, \Cref{theorem:bit-lower-bound-additive-reweighted-subgraph} can be extended
to the directed all-or-nothing splitting function $\splitdaon$,
to show a lower bound of
$\Omega(n^3 /\epsilon)$ bits.
For details see \Cref{sec:encoding-size-lower-bounds}.

Finally, we can also prove a space lower bound
for an arbitrary encoding of cuts in a directed hypergraph 
(arbitrary means that it need not represent a reweighted-subgraph sparsifier, see \Cref{sec:encoding-size-lower-bounds} for details).
This proof  provides an $\epsilon^{-1}$ factor improvement
over the trivial lower bound of $\Omega\left( n^2 \right)$ bits.
The proof combines the techniques from \Cref{theorem:bit-lower-bound-additive-reweighted-subgraph}
with a lower bound from \cite{OST22} on the number of edges in a reweighted-subgraph sparsifier. 
\begin{theorem}
  \label{theorem:directed-hypergraph-encoding-lower-bound}
  There exists a family of directed hypergraphs $\mathcal{H}$
  such that encoding a $(1+\epsilon)$-approximation of their cuts
  requires $\Omega(n^2/\epsilon)$ bits.
\end{theorem}

\paragraph{Techniques.} 
Our lower bound for the encoding size of reweighted-subgraph sparsifiers
(\Cref{theorem:bit-lower-bound-additive-reweighted-subgraph})
boils down to a counting argument on a large family of hypergraphs $\mathcal{H}$,
that have sufficiently different cut values and thus require distinct encodings.
We construct hypergraphs in this family $\mathcal{H}$
by partitioning the vertices into three parts $V,U,W$,
and adding hyperedges that contain vertices from all three parts. 
We first create hyperedges consisting of a large random subset of vertices from $V$;
this adds entropy that will differentiate between hypergraphs in $\mathcal{H}$.
We then augment each hyperedge with vertices from $U$,
where each hyperedge is defined by a word in the Hadamard code.
We use the structure of this code to show that
by making cut queries to a hypergraph $H\in\mathcal{H}$,
one can recover the random bits encoded in the adjacency matrix of $H$ induced on $V$.
We use $W$ to create an unsparsifiable hypergraph,
i.e., one where removing any hyperedge will violate the approximation guarantee.
Finally, every hyperedge on $V\cup U$ is combined with a hyperedge on $W$.

\subsection{Deformation Lower Bounds}
Our success in finding a small succinct representation for additive functions motivates searching for deformations of other splitting functions.

A similar problem, of approximating a submodular function
by functions of small support but over the uniform distribution 
(i.e., in average-case rather than worst-case),
has received significant attention \cite{FVK13, CKKL12, GHRU13, FK14, FV16},
and it is known that every submodular function $f:2^V\to [0,1]$
can be approximated within additive error $\epsilon$ 
using support size $O(\epsilon^{-2}\log \epsilon^{-1})$ \cite{FV16}.
We show (see \Cref{sec:deformation-lower-bounds}) that a similar result is unfortunately not possible in our setting (multiplicative error for worst-case approximation).

\begin{theorem}[Approximation Requires Large Support Size]
  \label{theorem:lower-bound-support-additive-functions}
  Let $g_e$ be an additive splitting function on a hyperedge $e$. 
  Then every $1.1$-approximation of $g_e$
  must have support size $p\geq \Omega(|e|/K)$.
\end{theorem}
  
\paragraph{Techniques.} 
The proof of \Cref{theorem:lower-bound-support-additive-functions} is based on \Cref{lemma:lower-bound-support-size},
a technical result that can be applied to many splitting functions.
The main idea is to examine a certain quantity $\delta_t$,
which is related the notion of curvature (of a submodular function).
The curvature is often used to parameterize approximation guarantees in maximization of submodular optimization \cite{CC84,vondrak2010submodularity}.
Intuitively, both the curvature and $\delta_t$ characterize the locality of the function, i.e., how much error is introduced by decomposing the function into smaller parts and summing them.
The main difference between the two quantities is that the curvature looks at the marginal contributions and $\delta_t$ characterizes the curvature of the union of two sets of size $t$.
Furthermore, in the approximation setting, a low worst-case curvature is desirable while for our proof it suffices that $\delta_t$ is high for many sets of size $t$.
Specifically, we show that if a constant fraction of pairs of subsets of size $t$ have constant positive $\delta_t$,
then $g_e$ cannot be approximated with support size smaller than $O(\delta_t^2 n/t)$.

By applying \Cref{lemma:lower-bound-support-size}, we obtain
lower bounds on the support size required to approximate several natural splitting functions,
as presented in~\Cref{table:summary-of-support-size-lower-bounds}.

\begin{table}
    \centering
    \begin{tabular}{|l|c|c|l|}
        \hline
        \textbf{Function Family} &
        \textbf{Example} &
        \textbf{Support Size} & 
        \textbf{} \\
        \hline
        additive functions 
        & $g_e(S) = \min(|S|,K)$ 
        & $\Omega\left( n/K \right)$
        & Lemma \ref{theorem:lower-bound-support-additive-functions}
        \\
        \hline
        polynomial
        & $g_e(S) = |S|^\alpha$ for constant $\alpha\in (0,1)$ 
        & $\Omega\left( n \right)$
        & Corollary \ref{corollary:lower-bound-polynomial}
        \\
        \hline
        logarithmic
        & $g_e(S) = \log(|S|+1)$ 
        & $\Omega\left( n \right)$
        & Corollary \ref{corollary:lower-bound-polylog}
        \\
        \hline
        cardinality based
        & $g_e(S) = f\left( |S| \right)$ for concave $f$ 
        & $\Omega\left( n/\mu_e^{1.5} \right)$
        & Corollary \ref{corollary:lower-bound-cardinality}
        \\
        \hline
        unweighted
        & $g_e(v) = 1$ for all $v\in V$ 
        & $\Omega\left( n/\mu_e^{3} \right)$
        & Corollary \ref{corollary:lower-bound-unweighted}
        \\
        \hline
    \end{tabular}
    \caption{Our lower bounds on the support size for several families of splitting functions.
      They are all obtained by applying \Cref{lemma:lower-bound-support-size},
      stated for simplicity for sufficiently small fixed $\epsilon>0$ and $|e|=n$.
    }
    \label{table:summary-of-support-size-lower-bounds}
\end{table}

\subsection{Related Work}
Submodular functions appear in many applications, and have been studied extensively in the literature.
In particular, the problem of finding a simple representation for submodular functions has been studied in several works.
An $O(\sqrt{n}\log n)$-approximation for monotone submodular functions by functions of the form $f(S)=\sqrt{\sum_{v\in S} c_v}$, where $c_v>0$ are weights for all $v\in V$, was obtained in \cite{GHIM09}.
A later result \cite{DDSSS13} showed the same approximation using coverage and budget-additive functions.
The same paper also provided a lower bound of $\Omega(n^{1/3} \log ^{-2} n)$ for approximating monotone submodular functions by coverage and budget additive.
Approximating the all-or-nothing splitting function on $n$ vertices
using hyperedges with the all-or-nothing function and with support size $r$ 
must incur approximation factor $\Omega(n/r)$ \cite[Section 2.3]{Pogrow17}. 

It was previously shown that every  symmetric cardinality-based splitting functions can be deformed into a sum of $|e|/2$ hyperedges with capped small-side splitting function,
while preserving the value of $g_e$ exactly \cite{VBK22}.
Subsequent work by the same authors~\cite{VBK21},
achieves a similar deformation but with $(1+\epsilon)$-approximation 
and using only $O(\epsilon^{-1} \log |e|)$ hyperedges.
Notice the difference from our work,
which focuses on an approximation with small support size.

\subsection{Concluding Remarks}
Our work provides several promising directions for future work.
We prove that all submodular hypergraph admit sparsifiers of polynomial size (\Cref{theorem:upper-bound-sparsification-all}),
leaving a gap of $\tOmega_{\epsilon}(n)$ between the upper and lower bounds.
We conjecture that submodular hypergraphs admit the same sparsification bounds
as (the special case of) directed hypergraphs.
\begin{conjecture}
  Every submodular hypergraph admits a $(1+\epsilon)$-sparsifier of size $O(\epsilon^{-2}n^2)$, which is in fact a reweighted sub-hypergraph.
\end{conjecture}

Notice that the known lower bound of $\Omega(n^2/\epsilon)$ is not tight with this conjecture, and improving it is an interesting open problem.
The main challenge in bridging the gap between our upper bound in \Cref{theorem:upper-bound-sparsification-all} and the conjecture
is the use of a union bound over all $2^n$ cuts.
This challenge was overcome in graph and hypergraph sparsification
by different methods,
such as cut counting \cite{benczur1996approximating,FHHP19, CKN20,KPS24},
a matrix Chernoff bound \cite{SS11},
and chaining which uses progressively finer discretizations \cite{BST19,KKTY21,OST22,JLLS23}.
Unfortunately, the matrix Chernoff bound is based on linear-algebra tools
that are clearly inapplicable to hypergraphs.
The cut-counting methods partition the cuts 
so that a union bound can be applied separately on each part; 
however these partitions rely on the binary nature of the all-or-nothing splitting function,
which seems challenging in the submodular hypergraph setting,
because the same $g_e$ can contribute very different values to different cuts $S\subseteq V$.
The chaining methods seem more promising,
especially the recent one~\cite{JLLS23} for all symmetric submodular functions,
in which the contribution of a single $g_e$ is not binary, 
although it seems to rely on the splitting functions being symmetric. 

In the sparsification setting,
we obtain smaller sparsifiers for several families (monotone and finite-spread),
however characterizing the optimal sparsifier size for each family remains open.
In the succinct-representation setting,
we found a useful deformation only for additive splitting functions (\Cref{theorem:additive-approximation}),
and it would be desirable to find deformations for more families.

Another interesting avenue is to find applications or connections to other problems.
For example, we show that \Cref{theorem:upper-bound-sparsification-finite-spread}
can be used to approximate a set-coverage function using a small ground set,
see \Cref{appendix:coverage} for details.
Another potential application is constructing succinct representations
for terminal cuts in a graph, see \Cref{appendix:terminal-cut-functions}.

\section{Polynomial-Size Sparsifiers for Submodular Hypergraphs}
\label{sec:polynomial-submodular-sparsifiers-proof}
This section proves \Cref{theorem:upper-bound-sparsification-all} and its improvement in the monotone case.
Our sparsification method is based on importance sampling,
where hyperedges are sampled with probability that is (at least)
proportional to their maximum relative contribution to any cut.
A standard choice, that was indeed used in~\cite{RY22},
is to sample every $e\in E$ with probability exactly proportional
to its importance, defined as
\[
    \sigma_e \coloneqq \max_{S\subseteq V}  \frac{g_e(S\cap e)}{\sum_{f\in E} g_f(S\cap f)}.
\]
The expected size of this sparsifier is proportional to the total importance $\sum_{e\in E} \sigma_e$,
which is non-trivial to bound
(e.g., naively replacing the maximization over ${S\subseteq V}$ by summation yields an exponential size bound).
An improved bound on the size of a sparsifier constructed in this manner is given in~\cite{RY22}, 
based on a quantity $B_H$ related to the polytopes of the splitting functions.
Unfortunately, this improved bound is still exponential for many families of splitting functions.

Our approach achieves a polynomial bound by using a different set of sampling probabilities and a different analysis.
Our main insight is that it suffices to consider only a few cuts.
Formally, define the minimum directed cut of $g_e$ between $(u,v)\in V\times V$ as
\begin{equation} \label{eq:DirectedCut}
    \minhecut{e}{u\to v} 
    \coloneqq 
    \min_{S\subseteq V: u\in S, v \not\in S} g_e(S\cap e)
    .
\end{equation}
Notice that we do not require $u,v\in e$; 
clearly, $\minhecut{e}{u\to v} = 0$ if $u\not\in e$,
but $\minhecut{e}{u\to v}$ can be positive if $v\not\in e$.
Our sampling probabilities are proportional to
\[
  \rho_e \coloneqq
  \sum_{(u,v)\in V\times V} \frac{ \minhecut{e}{u\to v} }{ \sum_{f\in E}\minhecut{f}{u\to v} }
  ,
\]
where by convention the fraction is equal to zero if the denominator (and thus also the numerator) is zero.
The proof follows by showing that $\rho_e\ge \sigma_e$, hence sampling every $e\in E$ with probability proportional to $\rho_e$ suffices to approximate the cuts,
and that the expected number of hyperedges in the sparsifier $O(\epsilon^{-2}n^3)$.
Since $\rho_e\ge \sigma_e$, our analysis implies that the same size bound
holds also for sampling with probabilities proportional to $\sigma_e$,
i.e., for the sparsifier of~\cite{RY22} but with our amplification factor $M=O(\epsilon^{-2}n)$. 

Finally, observe that the directed minimum cuts $\minhecut{e}{u\to v}$ can be computed in polynomial time using standard submodular minimization techniques \cite{mccormick2005submodular}.%
\footnote{In fact, computing an $O(1)$-approximation to $\rho_e$ would suffice,
  and this may be used to speed up the computation,
  at the cost of increasing the sparsifier size only by a constant factor.
}
In contrast, calculating $\sigma_e$ requires maximizing the ratio of two submodular functions, which is NP-hard.
In the monotone case, previous work had achieved a polynomial running time \cite{RY22,KZ23}.

\begin{proof}[Proof of Theorem \ref{theorem:upper-bound-sparsification-all}]
    Our construction of a quality $(1+\epsilon)$-sparsifier for $H$ uses the importance sampling method,
    where each hyperedge is sampled independently with probability $p_e$ that is defined below,
    and the splitting functions of every sampled hyperedge $d$ is scaled by factor $1/p_e$.

    We will use the following claim to bound cuts of $H$ by minimum directed cuts.
    Throughout, we denote $\bar{S} = V \setminus S$.
    \begin{claim}
        \label{claim:cut-bound-by-pairwise-flow}
        For every $e\in E$ and $S\subset V$,
        \[
            \max_{u\in S, v\in \bar{S}} \minhecut{e}{u\to v}
            \le 
            g_e(S\cap e)
            \le 
            \sum_{u\in S} \sum_{v\in \bar{S}} \minhecut{e}{u\to v}  
            .
        \]
    \end{claim}
    The proof of \Cref{claim:cut-bound-by-pairwise-flow} appears later.
    Intuitively, it is similar to bounding the capacity of a cut in a graph by the sum of maximum flows between each vertex from $S$ and each vertex from $\bar{S}$.
    We proceed assuming this claim, to show that $\rho_e\ge \sigma_e$.
    \begin{corollary}
        \label{corollary:bound-naive-importance}
        For every $e\in E$ and  $S\subseteq V$, we have $\rho_e \ge g_e(S\cap e)/\mintcut_H(S)$.
    \end{corollary}
    \begin{proof}
      By \Cref{claim:cut-bound-by-pairwise-flow},
      using both the upper bound and the lower bound on $g_e(\cdot)$, 
        \begin{equation*}
            \frac{g_e(S\cap e)}{\mintcut_H(S)}
            = \frac{g_e(S\cap e)}{\sum_{f\in E} g_f(S\cap f)}
            \le \sum_{u\in S, v\in \bar{S}} \frac{\minhecut{e}{u\to v}}{\sum_{f\in E} \minhecut{f}{u\to v}}
            \le \rho_e
            . 
        \end{equation*}
        Note that the first inequality holds even if $\mintcut_H(S)=0$, by our convention that if the denominator (and thus also numerator) is zero then the fraction is zero.
    \end{proof}

    For every hyperedge $e\in E$,
    set $\rho_e' \coloneqq g_e(e)/\sum_{f\in E} g_f(f)$ as the importance of the cuts that contain the entire hyperedge (the case $S=V$),
    and let $p_e \coloneqq \min (1, M(\rho_e+\rho_e'))$ 
    for a suitable parameter $M=O(\epsilon^{-2}n)$.
    Now sample every hyperedge $e\in E$ independently with probability $p_e$
    and rescale the splitting functions of every sampled hyperedge by factor ${1}/{p_e}$.
    Let $H'$ be the resulting hypergraph.

    We first prove that the number of hyperedges in the sparsifier $H'$ is $O(M n^2)$,
    which satisfies the claimed size bound by our choice of $M=O(\epsilon^{-2}n)$.
    Let $I_{e}$ be an indicator for the event that the hyperedge $e$ is sampled into $H'$.
    The expected number of sampled hyperedges is 
    \begin{align*}
        \Exp{\sum_{e\in E} I_{e}}{}
        = \sum_{e\in E} p_{e}
        &\le  M\sum_{e\in E} \left( 
        \frac{g_e(e)}{\sum_{f\in E} g_f(f)}
        +
        \sum_{(u,v) \in V\times V} \frac{\minhecut{e}{u\to v}}{\sum_{f\in E} \minhecut{f}{u\to v}} \right)
        \\
        &\le 
         M\left( 1+\sum_{(u,v) \in V\times V} \frac{\sum_{e\in E}\minhecut{e}{u\to v}}{\sum_{f\in E} \minhecut{f}{u\to v}}  \right)
        \le Mn^2
        ,
    \end{align*}
    where the second inequality follows by changing the order of summation
    and the last one is because $|V\times V| = n^2$,
    but we can exclude from the summation the case $u=v$ (as it contributes $0$ by our convention).
    By Markov's inequality, with high constant probability the sparsifier has at most $O(M n^2)$ hyperedges.
    
    Let us prove that the sparsifier $H'$ indeed approximates the cuts of $H$.
    Fix some $S\subseteq V$ and notice that 
    \begin{align*}
        \Exp{\mintcut_{H'}(S)}{}
        =\Exp{\sum_{e\in E} I_{e} \cdot \frac{1}{p_{e}} g_e(S\cap e)}{}
        = \sum_{e\in E} \frac{g_e(S\cap e)}{p_{e}} \cdot \Exp{I_{e}}{}
        =\sum_{e \in E}  g_e(S\cap e)
        = \mintcut_{H}(S)
        .
    \end{align*}
    Hence, the cut is preserved in expectation. 
    We shall now prove that the value of the cut is concentrated around its expectation.
    Let $Q_S=\left\{e\in E: p_{e}\in (0,1)\wedge g_e(S\cap e)>0 \right\}$ be the set of all hyperedges whose contribution to $\mintcut_{H'}(S)$ is random. 
    Furthermore, denote the maximum contribution of any such hyperedge to $\mintcut_{H'}(S)$ by $b\coloneqq\max_{e\in Q_S} p_e^{-1}g_e(S\cap e)$.
    By the Chernoff bound for bounded variables (\Cref{lemma:chernoff-bounded}),
    \begin{equation}
        \label{eq:chernoff-bound-general-submodular}
        \Probability{\mintcut_{H'}(S) 
        \not\in  (1\pm\epsilon) \cdot \mintcut_H(S)} 
        \le 2\cdot\exp\left(-\frac{\epsilon^2 \cdot \mintcut_H(S)}{b}\right)
        .
    \end{equation}
    We first analyze the special case $S=V$.
    Observe that if $\mintcut_H(V)=0$ then the cut is preserved trivially.
    Otherwise, note that $p_e \ge M\rho_e' = \frac{Mg_e(e)}{\sum_{f\in E} g_f(f)}$ and hence 
    \[
        b 
        = \max_{e\in Q_V} \frac{g_e(e)}{p_e} 
        \le\max_{e\in Q_V} g_e(e) \frac{\sum_{f\in E} g_f(f)}{M g_e(e)} 
        =  \frac{\mintcut_H(V)}{M}
        .
    \]
    Plugging this back into \Cref{eq:chernoff-bound-general-submodular}, we find $\Probability{\mintcut_{H'}(V) 
    \not\in  (1\pm\epsilon) \cdot \mintcut_H(V)} \le 2\cdot\exp\left(-\epsilon^2M \right)$.
    Now turning to the general case $S\subset V$, observe that by \Cref{corollary:bound-naive-importance}, $p_e \ge M g_e(S\cap e)/\mintcut_H(S)$. 
    Hence, we again obtain that
    \begin{align}
        b 
        \le \max_{e\in E} g_e(S\cap e) \cdot \frac{\mintcut_H(S)}{Mg_e(S\cap e)}
        = \frac{\mintcut_H(S)}{M}
        .
    \end{align}    
    Plugging this back into our concentration bound, \Cref{eq:chernoff-bound-general-submodular}, we get
    \[
        \Probability{\mintcut_{H'}(S) 
            \not\in  (1\pm\epsilon) \cdot \mintcut_H(S)} 
        \le 2\cdot\exp\left(-\epsilon^2 M\right)
        .
    \]
    Notice that this is the same probability as the case $S=V$.
    Setting $M \coloneqq c\cdot \epsilon^{-2}n$ for large enough but fixed $c>0$,
    we get that $\mintcut_{H'}(S)$ approximates $\mintcut_{H}(S)$ up to a $1\pm\epsilon$ factor with probability at least $1-2\exp(-cn)$.
    Applying a union bound over all $S\subseteq V$ we get that the sparsifier approximates all cuts simultaneously with probability at least $1-2\exp(-cn)\cdot 2^n\ge1-2\exp(-n)$.
    This completes the construction of a quality $1+\epsilon$ sparsifier for $H$ with $O(\epsilon^{-2} n^3)$ hyperedges.
    
    We now turn back to proving \Cref{claim:cut-bound-by-pairwise-flow}.
    \begin{proof}[Proof of \Cref{claim:cut-bound-by-pairwise-flow}]
        Fix some $e\in E$ and $S\subset V$. 
        For each directed minimum cut, let $P_e^{u\to v} \coloneqq \arg \min_{S\subseteq V: S\cap \{u,v\} =\{u\}} g_e(S)$ be some set $S\subseteq V$ attaining the minimum cut value (breaking ties arbitrarily).
        We need to show that 
    \begin{equation}
        \max_{u\in S, v\in \bar{S}} g_e(P^{u\to v}_e)
        \le g_e(S \cap e)
        \le  \sum_{u\in S}\sum_{v\in \bar{S}}  g_e(P^{u\to v}_e)
        .
        \label{eq:cut-bound-by-pairwise-flow}
    \end{equation}
        The lower bound is immediate because $g_e(P^{u\to v}_e)$ is a minimizer over the cuts separating $u$ from $v$. 
        For the upper bound, since $g_e$ is submodular and non-negative, 
        \[
            \forall A,B\subseteq e,
            \quad
            g_e(A)+g_e(B) 
            \ge g_e(A \cap B) + g_e(A \cup B) 
            \ge g_e(A \cap B)
            ,
        \]
        and similarly, $g_e(A)+g_e(B) \ge g_e(A \cup B)$.
        Using these two inequalities and summing over all $v\in \overline{S}$ and $u\in S$, we get
        \[
            \sum_{u\in S}\sum_{v\in \overline{S}} g_e(P^{u\to v}_e) 
            \ge \sum_{u \in S} g_e\left(\bigcap_{v\in \overline{S}} P^{u\to v}_e \right)
            \ge g_e\left(\bigcup_{u\in S} \bigcap_{v\in \overline{S}} P^{u\to v}_e \right)
            .
        \]
        To conclude the proof we show that $S\cap e = \bigcup_{u\in S} \bigcap_{v\in \overline{S}} P^{u\to v}_e$.
        For all $u\in S\cap e$ we have $\{u\} \subseteq \bigcap_{v\in \bar{S}} P^{u\to v}_e$, therefore $S\cap e\subseteq\bigcup_{u\in S} \bigcap_{v\in \overline{S}} P^{u\to v}_e$.
        In addition, for all $u\in S$ we have $\bigcap_{v\in \overline{S}} P^{u\to v}_e \subseteq S\cap e$ if $u\in e$ and $P^{u\to v}_e=\emptyset$ otherwise, therefore $S\cap e = \bigcup_{u\in S} \bigcap_{v\in \overline{S}} P^{u\to v}_e$.
        We conclude that \Cref{eq:cut-bound-by-pairwise-flow} holds.
        \end{proof}
    
    This completes the proof of \Cref{theorem:upper-bound-sparsification-all}.
\end{proof}

\subsection{Monotone Submodular Hypergraphs}
This section proves that every monotone submodular hypergraph admits a quality $(1+\epsilon)$-sparsifier of size $O(\epsilon^{-2}n^2)$.
\begin{theorem}
    \label{theorem:upper-bound-sparsification-monotone}
    Every hypergraph with monotone splitting functions admits a quality $(1+\epsilon)$-sparsifier of size $O(\epsilon^{-2}n^2)$, which is a reweighted sub-hypergraph.
\end{theorem}
  
The proof for the monotone case is similar to the general case.
However, since monotone splitting functions are more structured it suffices to examine the importance of all the singleton cuts for each hyperedge.
This results in smaller sampling probabilities and a better bound on the number of hyperedges in the sparsifier.
The proof utilizes the following well known property of monotone submodular functions.
\begin{claim}
    \label{claim:ratio-mincut-maxcut}
    Let $g_e:2^e\to\R_+$ be a monotone submodular splitting function. 
    Then
    \[
    \forall S\subseteq V,
    \qquad
    \max_{v\in S} g_e(\{v\}\cap e)
    \le    
    g_e(S\cap e) 
    \le \sum_{v\in S} g_e(\{v\}\cap e)
    .
    \]
\end{claim}
\begin{proof}
    The lower bound holds as $g_e$ is monotone.
    For the upper bound, since $g_e$ is submodular and non-negative,
    \begin{equation*}
        \sum_{v\in S} g_e(\{v\}\cap e)
        \ge g_e\left(\bigcup_{v \in S} \{v\} \cap e\right)
        = g_e(S\cap e) .
    \qedhere
    \end{equation*}
\end{proof}
Similarly to the general case, our over sampling probabilities are proportional to 
\begin{equation*}
    \rho_e =  \sum_{v\in V} \frac{g_e(\{v\}\cap e)}{\sum_{f\in E} g_f(\{v\}\cap f)}
    .
\end{equation*}
The following corollary shows that $\rho_e\ge \sigma_e$.
This implies that sampling every $e\in E$ with probability proportional to $\rho_e$ suffices to approximate the cuts of $H$, in the same manner as in the general case.
\begin{corollary}
    \label{corollary:bound-naive-importance-monotone}
    For every $e\in E$ and $S\subseteq V$,
    we have $\rho_e \ge g_e(S\cap e)/\mintcut_H(S)$.
\end{corollary}
\begin{proof}
    Observe that by \Cref{claim:ratio-mincut-maxcut},
    \begin{equation*}
        \frac{g_e(S\cap e)}{\mintcut_H(S)}
        = \frac{g_e(S\cap e)}{\sum_{f\in E} g_f(S\cap f)}
        \le \sum_{v\in S} \frac{g_e(\{v\}\cap e)}{\sum_{f\in E} g_f(\{v\}\cap f)}
        \le \rho_e
        .
    \end{equation*}
    Notice that the first inequality is well-defined by the convention that if the denominator (and thus also the numerator) is zero then the fraction is zero.
\end{proof}

We now turn to proving \Cref{theorem:upper-bound-sparsification-monotone}
\begin{proof}[Proof of \Cref{theorem:upper-bound-sparsification-monotone}]
    To construct the sparsifier $H'$, sample each hyperedge with probability $p_e = \min (1, M\cdot \rho_e)$ for a suitable parameter $M=O(\epsilon^{-2}n)$.
    Then, reweigh every sampled hyperedge by factor $p_e^{-1}$.
    The proof that $H'$ is with high probability a $(1+\epsilon)$-sparsifier
    is similar to the general case because $\rho_e\ge \sigma_e$,
    and we omit it.

    To bound the number of hyperedges in the sparsifier, let $I_{e}$ be an indicator for the event that the hyperedge $e$ is sampled into $H'$.
    Then the expected number of sampled hyperedges is,
    \begin{align*}
        \Exp{\sum_{e\in E} I_e}{} 
        = \sum_{e\in E} p_{e}
        \le M \sum_{e\in E} \sum_{v\in V} \frac{g_e(\{v\} \cap e)}{\sum_{f\in E} g_f(\{v\}\cap f)}
        \le M \sum_{v\in V} \sum_{e\in E} \frac{g_e(\{v\}\cap e)}{\sum_{f\in E} g_f(\{v\}\cap f)}
        \le M n
        ,
    \end{align*}
    where the second inequality is from changing the order of summation.
    Hence, by Markov's inequality we find that with high constant probability the size of the sparsifier is at most $O(Mn)=O(\epsilon^{-2}n^2)$.
    This concludes the proof.
\end{proof}

\section{Sparsifiers for Finite-Spread Splitting Functions}
\label{sec:finite-spread-sparsifiers-proof}
This section provides a construction of sparsifiers for finite-spread splitting functions (\Cref{theorem:upper-bound-sparsification-finite-spread}).
Our construction is based on the method presented in \cite{CKN20} for constructing cut sparsifiers for the all-or-nothing splitting function.
The main argument is that by approximating every hyperedge up to a factor of $\mu_H$ as the all-or-nothing splitting function, we can follow the algorithm and proof of \cite{CKN20}.
This approximation only holds when the functions $g_e$ have finite spread.
The main difference is that we need to account for hyperedges contributing different amounts to different cuts.
We show that by oversampling hyperedges at a rate higher by a $\mu_H$ factor, we can adjust the Chernoff bounds, and then the rest of the proof follows using the all-or-nothing approximation.
Throughout the proof we assume that $\min_{S\subseteq e:g_e(S)\ne 0} g_e(S) =1$ for all $e\in E$.
We can make this assumption by recalling that we limited our discussion to splitting functions with integral values, and observing that a splitting function $g_e$ with a higher minimal value can be divided into multiple functions with minimal value $1$ without affecting the cuts of $H$.

We begin by presenting the relevant definitions and results from the existing literature.
\begin{definition}
    Let $G=(V,F,w)$ be a weighted graph. A \emph{$k$-strong component} in $G$ is a maximal vertex induced subgraph such that the minimum cut in the component is $k$.
\end{definition}
\begin{lemma}[\cite{benczur1996approximating}]
    Given a weighted graph $G=(V,F,w)$ and some $k>0$, the $k$-strong components of $G$ partition $V$. For every $k'>k$ the $k'$-strong components are a refinement of the $k$-strong components.
\end{lemma}	
\begin{definition}
    In a weighted graph $G=(V,F,w)$, the \emph{strength of an edge} $f\in F$, denoted by $k_f$, is the maximal $k>0$ such that $f$ is contained in a $k$-strong component.
\end{definition}
\begin{claim}[Corollary 4.9 in \cite{BK15}]
    In every weighted graph $G$ on $n$ vertices, there are at most $n-1$ distinct values of edge strengths.
    \label{claim:number-of-strength-values}
\end{claim}

Following the proof in \cite{CKN20}, the sampling probabilities of the hyperedges of $H=(V,E,\g)$ are determined by an auxiliary weighted graph $G$,
where every hyperedge $e\in E$ is represented by a weighted clique $F_e$ in $G$ (with perhaps some weights being zero).
Observe that $G$ may have many parallel edges between the same pair of vertices, each induced by a different hyperedge.
Denote the set of edges with positive weight in $F_e$ by $F_e^+$.
The precise construction of the auxiliary graph $G$ is described in \cite{CKN20}.
Define the \emph{hyperedge strengths} as $\kappa_e \coloneqq \min_{f\in F_e} k_f$ and $\kappa_e^{\max} \coloneqq \max_{f\in F_e^+ } k_f$.
\begin{theorem}[Theorem 3 in \cite{CKN20}]
    Let $H=(V,E,\g)$ be a hypergraph with a finite-spread splitting function $g:2^V\to \R_+$.
    For every integer $\gamma \ge 2$ there exists an assignment of weights to the edges of $F_e$ such that in the resulting $G$
    \begin{itemize}
        \item $\sum_{f\in F_e} w(f) = 1$;
        \item $\frac{\kappa_e^{\max}}{\kappa_e} \le \gamma$.
    \end{itemize}
    \label{theorem:hyperedge-strengths}
\end{theorem}

We now turn to proving Theorem \ref{theorem:upper-bound-sparsification-finite-spread}.
Throughout the proof we will partition the hyperedges into two sets, those with $g_e(e)=0$ and those with $g_e(e)>0$. 
We sparsify each subgraph independently and then the union of the two sparsifiers to obtain the sparsifier for $H$.
Note that if each sparsifier approximates its hyperedges with quality $1+\epsilon$, then their union $(1+\epsilon)$-approximates all the cuts of $H$ by the additivity of the cuts; this increases the size of the sparsifier by at most a factor of $2$.
The sampling process for both sets is identical and hence from now on we assume that either all the hyperedges of $H$ have  $g_e(e)=0$ or all have $g_e(e)>0$.

Start by applying Theorem \ref{theorem:hyperedge-strengths} with $\gamma=2$ to obtain the auxiliary graph $G$ and define the strengths of the hyperedges of $H$ accordingly.
Let $H'$ be a sparsifier constructed by sampling every hyperedge $e\in E$ with probability $p_e = \min\left\{1, \rho/\kappa_e \right\}$, 
for $\rho = \epsilon^{-2}t\mu_H\gamma^2 \log n$ with constant $t >0$ to be determined later. 
For each sampled edge we assign a rescaled splitting function $g_e' = p_e^{-1} g_e$.

By Claim 6 in \cite{CKN20} the size of the sparsifier resulting from this sampling method is at most $O\left(\rho \gamma n \right) = O\left( \epsilon^{-2} n t \mu_H\gamma^3 \log n \right)$.

Hence, it remains to show that $H'$ approximates the cuts of $H$.
The proof is based on partitioning the hyperedges of $G$ based on their strength, $\kappa_e$, and showing that the additive error for each set of hyperedges is small.
Define the following sets.
Let $E_{\ge i} = \left\{e\in E: \kappa_e \ge \rho 2^i \right\}$ be the set of all hyperedges in $H$ with strength at least $\rho 2^i$,
and let the set of all hyperedges with strength in $\left[\rho 2^i, \rho 2^{i+1}\right)$ be $E_{i}=E_{\ge i} \setminus E_{\ge i+1 }$.
Similarly, let $F_{\ge i} = \left\{f\in F^+: k_f \ge \rho 2^i \right\}$ be the set of all edges in $G$ of strength at least $\rho 2^i$.
Finally, let $E_{\ge i}^{\max} = \left\{e\in E: \kappa_e^{\max} \ge \rho 2^i \right\}$ be the set of all hyperedges with maximum strength at least $\rho 2^i$.

Let $Y\subseteq E$ be a subset of the hyperedges of $H$ and let $\beta \in \R^{\left|Y\right|}_+$.
Denote by $H\left[\beta Y \right]$ the sub-hypergraph with the hyperedges in $Y$, where for each $e\in Y$ the splitting function $g_e$ is scaled by $\beta_e$.
Similarly, for $Y\subseteq F$, and $\beta \in \R^{\left|Y\right|}_+$, let $G\left[\beta Y\right]$ be the subgraph of $G$ with the edges in $Y$ and the weight of each $f\in Y$ scaled by $\beta_f$.

We now define certain subsets that will be used to bound the error.
For every $i\in \N$, define a relative weighting function $\beta^{i} \in \R_+^E$ where $\beta_e^i= 2^{i-j}$ for all hyperedges $e\in E$ (or similarly for edges in $G$) with strength in $\left[\rho 2^j, \rho 2^{j+1}\right)$.
Observe that this function modifies the weight of each hyperedge (or edge) according to its strength.
Finally, let the random vector $\hat{\beta}^i\in \R_+^{E_i}$ be the scaling obtained from the sampling process limited to $E_i$;
i.e. $\hat{\beta}^i_e = p_e^{-1}$ if $e$ is sampled into $H'$ and $\hat{\beta}^i_e = 0$ otherwise.
Observe that $H' =  H\left[\bigcup_i \hat{\beta}^i E_i \right]$.

Our main technical lemma bounds the additive error between the value of every cut $S$ on $\hat{\beta}^i E_i$ and its expectation.
Note that the lemma does not immediately imply that $H'$ is a quality $q=1+\epsilon$ sparsifier for $H$ as the error bound might be larger than $\epsilon \cdot \mintcut_{H[E_i]}(S)$.
However, we will show that the sum of all these error bounds is small compared to the value of the cut.

\begin{lemma}
    Fix an integer $i>0$.
     Then with probability at least $1-8n^{-2}$,
    \begin{equation}
        \forall S \subseteq V,
        \quad
        \left|
            \mintcut_{H[\hat{\beta}^i E_i]}\left(S\right) - 
            \mintcut_{H[E_i]}\left(S\right)
        \right|
        \le \frac{\epsilon}{\gamma}  \mintcut_{H[\beta^{i} E_{\ge i}^{\max}]}\left(S\right)
        .
        \label{eq:lemma-bound-error}
    \end{equation}
    \label{lemma:bound-set-partition-error}
\end{lemma}

We first show that Lemma \ref{lemma:bound-set-partition-error} implies Theorem \ref{theorem:upper-bound-sparsification-finite-spread}.
\begin{proof}[Proof of Theorem \ref{theorem:upper-bound-sparsification-finite-spread}]
    Sample the hyperedges of $H$ into $H'$ using the sampling method described above.
    Using the union bound, we bound the probability that the event in Equation (\ref{eq:lemma-bound-error}) holds for all values of $i$ whose corresponding hyperedge set $E_i$ is non-empty.
    By Claim \ref{claim:number-of-strength-values} there are at most $n-1$ such values.
    Hence, with probability at least $1-8n^{-1}$ Equation (\ref{eq:lemma-bound-error}) holds for all such $i$.
    
    We now show that if  \Cref{eq:lemma-bound-error} holds for all such $i$, then the sparsifier has quality $q=1+\epsilon$.
    Observe that for every $S\subseteq V$,
    \begin{align*}
        \left|
            \mintcut_{H[\hat{\beta}^i  E_i]}\left(S\right) - 
            \mintcut_{H[E_i]} \left(S\right)
        \right|
        &\le \frac{\epsilon}{\gamma}  \mintcut_{H[\beta^{i} E_{\ge i}^{\max}]}\left(S\right)
        \\
        &\le 
        \frac{\epsilon}{\gamma}\sum_{j\ge i - \log \gamma}  \mintcut_{H\left[ \beta^{i}E_j\right]}\left(S\right)
        =
        \frac{\epsilon}{\gamma}\sum_{j\ge i - \log \gamma}  2^{i-j} \cdot \mintcut_{H\left[E_j\right]}\left(S\right)
        ,
    \end{align*}
    where the second inequality is since $E_{\ge i}^{\max} \subseteq E_{
    \ge i -\log \gamma}$ (this holds since $\frac{\kappa^{\max}_e}{\kappa_e}\le \gamma$).
    Summing over all $i$,
    \begin{align*}
        \left|
        \mintcut_{H'} \left(S\right) - 
        \mintcut_{H} \left(S\right)\right|
        & =
        \sum_{i\ge 0} \left|\mintcut_{H[\hat{\beta}^i  E_i]}\left(S\right) - 
        \mintcut_{H[E_i]} \left(S\right)
        \right|
        \\
        &\le 
        \frac{\epsilon}{\gamma}\sum_{i\ge 0} \sum_{j\ge i - \log \gamma}  2^{i-j} \cdot \mintcut_{H\left[E_j\right]}\left(S\right)
        \\
        &=
        \frac{\epsilon}{\gamma} \sum_{j\ge 0} \mintcut_{H\left[E_j\right]}\left(S\right) \sum_{i \le j+ \log \gamma}  2^{i-j}
        \le 2\epsilon \sum_{j\ge 0} \mintcut_{H\left[E_j\right]}\left(S\right)
        = 2\epsilon \mintcut_{H} \left(S\right)
        ,
    \end{align*}
    where the second inequality is by $\sum_{i \le j+ \log \gamma}  2^{i-j} \le 2\gamma$. 
    Hence, with probability $1-8n^{-1}$ we have that $H'$ is a cut sparsifier of quality $q=1+\epsilon$ for $H$.
\end{proof}

To conclude the proof, we still need to show Lemma \ref{lemma:bound-set-partition-error}.

\subsection{Proof of Lemma \ref{lemma:bound-set-partition-error}}
To prove the lemma, we first provide several definitions and claims.
Define the \emph{sample set of level $i$} to be $\hat{E}_i = \hat{\beta}^i E_i + \beta^{i} E_{\ge i+1}$.
Since the contribution of each hyperedge to a cut is additive,
\[
    \mintcut_{H[\hat{E}_i]}\left(S\right) - 
    \mintcut_{H[\beta^{i} E_{\ge i}]}\left(S\right)
    = \mintcut_{H[\hat{\beta}^iE_i]}\left(S\right) -
    \mintcut_{H[E_{i}]}\left(S\right)
    ,
\]
notice that in the last term we omitted $\beta^i$ as for every $e\in E_i$ we have $\beta^i_e = 2^0$.
Hence, showing
\begin{equation}
    \left| 
        \mintcut_{H[\hat{E}_i]}\left(S\right) - 
        \mintcut_{H[\beta^{i} E_{\ge i}]}\left(S\right)
    \right|
    \le \frac{\epsilon}{\gamma} \mintcut_{H[\beta^{i} E_{\ge i}^{\max}]}\left(S\right)
    \label{eq:revised-lemma-bound-error}
\end{equation}
would suffice to prove the lemma.
The following claims outline useful properties of $G\left[F_{\ge i}\right]$.
\begin{claim}[Claim 8 from \cite{CKN20}]
    For any $e\in E_{\ge i}$, the entire vertex set $F_e$ belongs to the same connected component in $G\left[F_{\ge i}\right]$.
    \label{claim:same-connected-graph-hyper}
\end{claim}
\begin{claim}[Claim 9 from \cite{CKN20}]
Let $A_G$ be a connected component in $G\left[F_{\ge i}\right]$. Then the minimum cut size of $A_G$ is at least $\rho 2^i$.
\label{claim:minimum-cut-size-component}
\end{claim}

Our proof will also use the following cut-counting lemma.
\begin{lemma}[Corollary 8.2 in \cite{Karger93}]
    Given a weighted graph $G=(V,F,w)$ with minimum cut size $c>0$, for all integers $\alpha \ge 1$, the number of cuts of the graph of weight at most $\alpha c$ is at most $|V|^{2\alpha}$.
    \label{lemma:cut-counting}
\end{lemma}

To use the cut counting lemma, we show that the cuts of $H\left[\beta^{i}E_{\ge i}^{\max} \right]$ are bounded from below by the cuts of $G\left[\beta^{i} F_{\ge i}\right]$.
To do that, we use the following from \cite{CKN20}.
\begin{claim}[Claim 7 from \cite{CKN20}]
    Let $\bar{H}=(V,E,\bar{\g})$ be a hypergraph with the same vertices and hyperedges as $H$, but with $\bar{g}_e:2^e\to \R_+$ being the all-or-nothing splitting function, i.e.,
    \[
        \bar{g}_e(S) = 
        \begin{cases}
            0, & \text{if } |S\cap e| = 0 \text{ or } |S\cap e| = |e|
            ;
            \\
            1, & \text{otherwise} 
            .
        \end{cases}
    \]
    Then for all $i\ge 0$,
    \[
        \forall S\subseteq V,
        \quad
        \mintcut_{\bar{H}\left[\beta^{i}E_{\ge i}^{\max} \right]}\left(S\right)
        \ge \mintcut_{G\left[\beta^{i} F_{\ge i}\right]}\left(S\right)
        .
    \]
    \label{claim:all-or-nothing-cut-bound}
\end{claim}
We then obtain the desired lower bound as an immediate corollary.
This corollary is where we use our assumption that $g_e$ has finite spread, otherwise the cuts of $H$ are not bounded by the cuts of $G$, and we cannot use the cut counting lemma (\Cref{lemma:cut-counting}).
\begin{corollary}
    For all $i\ge 0$, 
    \[
        \forall S\subseteq V,
        \quad
        \mintcut_{H\left[\beta^{i}E_{\ge i}^{\max} \right]}\left(S\right)
        \ge \mintcut_{G\left[\beta^{i} F_{\ge i}\right]}\left(S\right)
        .
    \]
    Furthermore, if $g_e(e) >0$ for all $e\in E$ then for every connected component $A_G\subseteq V$ of $G\left[\beta^{i} F_{\ge i}\right]$, the cut $A_G$ on $H\left[\beta^{i}E_{\ge i}^{\max} \right]$ is at least the minimum cut of $G\left[\beta^{i} F_{\ge i}\right]$ in the component $A_G$.
    Formally,
    \[
        \mintcut_{H\left[\beta^{i}E_{\ge i}^{\max} \right]}\left(A_G \right)
        \ge \min_{T\subset A_G, T\ne \emptyset} \mintcut_{G\left[\beta^{i} F_{\ge i}\right]}\left(T\right)
        .
    \]
    \label{corollary:bound-cut-from-below}
\end{corollary}
\begin{proof}
    We begin by showing the first part of the corollary.
    Observe that since $g_e$ has finite spread, whenever $|e|>|S\cap e|\ge 1$ we have $g_e(S\cap e)\ge 1=\bar{g}_e(S\cap e)$.
    In addition, $g_e(e) \ge 0 = \bar{g}_e(e)$.
    Hence, for all $e\in E$ and $S\subseteq e$, $g_e(S\cap e) \ge \bar{g}_e (S\cap e)$.
    Therefore, for every cut $S\subseteq V$,
    \[
        \mintcut_{H\left[\beta^{i}E_{\ge i}^{\max} \right]}\left(S\right)
        \ge
        \mintcut_{\bar{H}\left[\beta^{i}E_{\ge i}^{\max} \right]}\left(S\right)
        \ge 
        \mintcut_{G\left[\beta^{i} F_{\ge i}\right]}\left(S\right)
        .
    \]

    We now turn to the second part of the corollary.
    Fix some connected component $A_G$ of $G\left[\beta^{i} F_{\ge i}\right]$ and denote some set achieving the minimum cut over $A_G$ by $T^* \subset A_G, T^* \ne \emptyset$.
    Note that the value of the cut $T^*$ is strictly positive since $A_G$ is a connected component.
    Notice that for all $S$ such that $S\cap e\ne \emptyset$ we have $g_e(S\cap e)\ge 1 \ge \bar{g}_e(S\cap e)$ by our assumption.
    Hence,
    \begin{align*}
        \mintcut_{H\left[\beta^{i}E_{\ge i}^{\max} \right]}(A_G)
        &=
        \sum_{e\in E_{\ge i}^{\max}} \beta^{i}_e g_e(A_G \cap e)
        \ge \sum_{e\in E_{\ge i}^{\max}} \beta^{i}_e \indic{A_G\cap e \ne \emptyset}
        \ge \sum_{e\in E_{\ge i}^{\max}} \beta^{i}_e \indic{T^* \cap e \ne \emptyset}
        \\
        & \ge
         \sum_{e\in E_{\ge i}^{\max}} \beta^{i}_e \bar{g}_e(T^*)
        = \mintcut_{\bar{H}\left[\beta^{i}E_{\ge i}^{\max} \right]}(T^*)
        \ge \mintcut_{G\left[\beta^{i} F_{\ge i}\right]}(T^*)
        ,
    \end{align*}
    where the second inequality is since $T^* \subseteq A_G$,
    the third is by the definition of the all-or-nothing splitting function,
    and the last one is by \Cref{claim:all-or-nothing-cut-bound}.
\end{proof}

To prove Lemma \ref{lemma:bound-set-partition-error}, we bound the error of each connected component of $H\left[E_{\ge i}\right]$ separately.
For each connected component, we bound the error by a term proportional to  $\mintcut_{H\left[\beta^{i}E_{\ge i}^{\max} \right]}\left(S\right)$, which by Corollary \ref{corollary:bound-cut-from-below} bounds the cut in $G\left[\beta^{i}F_{\ge i}\right]$ from above.
This allows us to utilize the cut-counting-lemma (Lemma \ref{lemma:cut-counting}) to bound the number of small cuts.

\begin{proof}[Proof of Lemma \ref{lemma:bound-set-partition-error}]
    Fix some $i\ge 0$ and let $A_G=(V_{A_G},F_{A_G})$ be some connected component in the subgraph $G\left[F_{\ge i}\right]$.
    Note that if for all $e\in E$ we have $g_e(e)=0$ then the error for the cut $S=V_{A_G}$ is $0$.
    By Claim \ref{claim:same-connected-graph-hyper} the hypergraph induced by the vertices $V_{A_G}$ is a proper sub-hypergraph of $H[E_{\ge i}]$, as every hyperedge in $E_{\ge i}$ is incident only to the vertices in $V_{A_G}$.
    Denote this connected component by $H_{A_G}$.
    Note, that every cut with value $0$ is preserved trivially by the sparsifier.

    Fix some cut $S\subseteq V_{A_G}$ such that $\mintcut_{H_{A_G}}(S)>0$.
    We will start by bounding the probability of $H_{A_G}[\hat{E}_i]$ creating a large error,
    \begin{equation}
        \label{eq:revised-lemma-bound-error-component}
        \left| 
            \mintcut_{H_{A_G}[\hat{E}_i]}\left(S\right) - 
            \mintcut_{H_{A_G}[\beta^{i} E_{\ge i}]}\left(S\right)
        \right|
        \ge \frac{\epsilon}{\gamma} \mintcut_{H_{A_G}[\beta^{i} E_{\ge i}^{\max}]}\left(S\right)
        .
    \end{equation}		
    Notice that,
    \[
        \Exp{\mintcut_{H_{A_G}[\hat{E}_i]}\left(S\right)}{}
        = \mintcut_{H_{A_G}[\beta^{i} E_{\ge i}]}\left(S\right)
        .
    \]
    Denote 
    \[
        \theta 
        \coloneqq \frac{\mintcut_{H_{A_G}[\beta_{i} E_{\ge i}^{\max}]}\left(S\right)}{\mintcut_{H_{A_G}[\beta^{i} E_{\ge i}]}\left(S\right)} \ge 1
        .		
    \]
    Observe that by the Chernoff bound when $\theta\epsilon/\gamma \ge 1$ the probability of large deviation is at the most the probability of large deviation in the case $\theta\epsilon/\gamma \le 1$. 
    Hence, we assume henceforth that $\theta\epsilon/\gamma \le 1$.
    By the Chernoff bound (Lemma \ref{lemma:chernoff-bounded}),
    \begin{align*}
        &\Probability{
            \left|
            \mintcut_{H_{A_G}[\hat{E}_i]}\left(S\right) - 
            \mintcut_{H_{A_G}[\beta^{i} E_{\ge i}]}\left(S\right)
            \right|
            \ge \frac{\epsilon}{\gamma}  \mintcut_{H_{A_G}[\beta_{i} E_{\ge i}^{\max}]}\left(S\right)
        }
        \\
        &=
        \Probability{
            \left|
            \mintcut_{H_{A_G}[\hat{E}_i]}\left(S\right) - 
            \mintcut_{H_{A_G}[\beta^{i} E_{\ge i}]}\left(S\right)
            \right|
            \ge \theta\frac{\epsilon}{\gamma}  \mintcut_{H_{A_G}[\beta_{i} E_{\ge i} ]}\left(S\right)
        }
        \\
        &\le 
        2\exp\left(
        -\frac{\epsilon^2  \theta^2\mintcut_{H_{A_G}[\beta^{i} E_{\ge i}]}\left(S\right)}{3\gamma^2 r}
        \right)
        \le
        2\exp\left(
        -\frac{\epsilon^2 \mintcut_{H_{A_G}[\beta_{i} E_{\ge i}^{\max}]}\left(S\right)}{3\gamma^2 r}
        \right)
        ,
    \end{align*}
    where $r$ is the maximum possible contribution of any sampled hyperedge to the cut.
    The second inequality is from $\theta \ge 1$.
    Observe that $r$ is bounded by
    \[
        r = \max_{e\in E, S\subseteq V} \frac{1}{p_e} g_e(S) \le \frac{\kappa_e}{\rho} \mu_H
        .
    \]
    Since all the sampled edges $e$ are in $E_i$, all of them have strength in $\kappa_e \in \left[\rho 2^i, \rho 2^{i+1 }\right)$.
    This implies that $r \le 2^{i+1} \mu_H$.
    Let $x=\min_{T\subseteq V_{A_G}}\mintcut_{A_G}(T)$ be the value of a minimum cut of $A_G$.
    We begin by lower bounding the value of the cut $S$ in $H_{A_G}[\beta^{i} E_{\ge i}^{\max}]$.
    By Corollary \ref{corollary:bound-cut-from-below}, if $S\ne V_{A_G}$ then $\mintcut_{H_{A_G}[\beta^{i} E_{\ge i}^{\max}]}(S) \ge \mintcut_{A_G}(S)$.
    Otherwise, by the same corollary we have that $\mintcut_{H_{A_G}[\beta^{i} E_{\ge i}^{\max}]}(S) \ge x$.
    Combining these and recalling $\mintcut_{A_G}(S)=0$ if $S=V_{A_G}$ we find that $\mintcut_{H_{A_G}[\beta^{i} E_{\ge i}^{\max}]}(S) \ge \max(\mintcut_{A_G}(S),x)$.

    Let $\alpha = \max (\mintcut_{A_G}(S),x)/x$ be the ratio between the lower bound on $\mintcut_{H_{A_G}[\beta^{i} E_{\ge i}^{\max}]}(S)$ and $x$.
    By \Cref{claim:minimum-cut-size-component}, we have that $x \ge \rho 2^i$.
    Therefore,
    \[
        \mintcut_{H_{A_G}[\beta_{i} E_{\ge i}^{\max}]}(S)
        \ge \alpha x 
        \ge \alpha \rho 2^i
        .
    \]
    Plugging in these two bounds 
    
    \begin{align*}
            &\Probability{
            \left|
            \mintcut_{H_{A_G}[\hat{E}_i]}\left(S\right) - 
            \mintcut_{H_{A_G}[\beta^{i} E_{\ge i}]}\left(S\right)
            \right|
            \ge \frac{\epsilon}{\gamma}  \mintcut_{H_{A_G}[\beta^{i} E_{\ge i}^{\max}]}\left(S\right)
            }
            \\
            &\le 
            2\exp\left(
            -\frac{\epsilon^2 \alpha \rho 2^i }{3\gamma^2 \mu_H 2^{i+1}}
            \right)
            = 
            2\exp\left(
            -\frac{\epsilon^2 \alpha}{6\gamma^2 \mu_H} \frac{t\mu_H\gamma^2 \ln n}{\epsilon^2}
            \right)
            =
            2\exp\left(
            -\frac{\alpha t \ln n }{6}
            \right)
            .
    \end{align*}
    Therefore, the event occurs with probability at most $2n^{\frac{-\alpha t}{6}}$.

    By the cut-counting lemma (\Cref{lemma:cut-counting}), there are at most $|V_{A_G}|^{2\alpha}$ cuts of value at most $\alpha x$ in $A_G$ which are different from $\emptyset,V_{A_G}$.
    Counting also the cut $S=V_{A_G}$, there are at most $|V_{A_G}|^{2\alpha}+1\le 2|V_{A_G}|^{2\alpha}$ such cuts in total.
    Hence, the probability that \Cref{eq:revised-lemma-bound-error-component} occurs for any cut is at most,
    \[
        \sum_{\alpha= 1}^{\infty} 4|V_{A_G}|^{2\alpha\left(\frac{-t}{6}+2\right)}
        \le 8|V_{A_G}|^{\frac{-t}{6}+1}
        .
    \]
    Using a union bound over all connected components we find that the probability that \Cref{eq:revised-lemma-bound-error-component} occurs for any cut in any connected component $A_{G_i}$ is at most
    \[
        \sum_{\alpha= 1}^{\infty} \sum_j 4|V_{A_{G_j}}|^{2\alpha\left(\frac{-t}{6}+2\right)}
        \le 8n^{\frac{-t}{6}+2}
        ,
    \]
    where inequality is since $\sum_j |V_{A_{G_j}}| \le n$.
    Choosing $t>24$ we find that the error term satisfies the requisite bound for all cuts simultaneously with probability at least $1-8n^{-2}$.
    Therefore, \Cref{eq:revised-lemma-bound-error} holds with high probability.
    This completes the proof of \Cref{lemma:bound-set-partition-error}.
\end{proof}

\section{Deformation of Additive Splitting Functions}
\label{sec:deformation-additive-splitting}
This section proves \Cref{theorem:additive-approximation} by showing that every hyperedge with an additive splitting function $g_e:2^e\to \R$ can be decomposed into multiple hyperedges, each with support size at most $ O (\epsilon^{-2}K^{-1}|e| \log |e|) $.
Our proof is based on sampling, namely, we approximate $g_e$ by a sum of hyperedges $\left\{ e_i \right\}_i$ with additive splitting functions $\left\{ g_{e_i} \right\}_i$, such that each $e_i \subseteq e$ is constructed by a uniform independent sample of the vertices of $e$.
The main challenge in the proof is showing that for every $S\subseteq e$, we have $\Exp{g_{e_i}(S\cap e_i)}{}\in (1\pm \epsilon) g_e(S)$.
We also prove that \Cref{theorem:additive-approximation} implies a succinct representation of hypergraphs with additive splitting functions (\Cref{corollary:additive-small-representation}).
We begin by presenting a more detailed version of \Cref{theorem:additive-approximation}.

\begin{lemma}[Detailed Statement of \Cref{theorem:additive-approximation}]
    \label{lemma:deformation-additive}
    Let $e$ be a hyperedge with an additive splitting function $g_e$ with parameter $K\le |e|$.
    Then, for every $\epsilon\in(0,1)$ there exists a deformation of $e$ into $N=O\left(\epsilon^{-2}|e|^2\right)$ hyperedges $\{e_i\}_i$ with additive splitting functions$\{g_{e_i}\}_i$ with parameter $K'$, such that
    \[
        \forall S\subseteq V,
        \quad
        \sum_{i=1}^N g_{e_i}(S\cap e_i) \in \left( 1\pm \epsilon \right)g_e(S)
        .
    \]
    In addition, the cardinality of each hyperedge satisfies $|e_i|\le O\left( \epsilon^{-2}(|e|/K)\log |e| \right)$ and its spread is $\mu_{e_i}\le O(\epsilon^{-2}\log |e|)$.
    Moreover, this decomposition can be found by a randomized sampling process with success probability at least $1-2|e|^{-4}$.
\end{lemma}

Using the above decomposition we obtain a succinct representation of hypergraphs with additive splitting functions.
\begin{proof}[Proof of \Cref{corollary:additive-small-representation}]
    Let $H=(V,E,\g)$ be a hypergraph such that all $g_e\in \g$ are additive splitting functions, each with parameter $K_e$.
    Denote $\hat{K}=\min_{e\in E} K_e/|e|$.
    Applying \Cref{theorem:additive-approximation} on all the hyperedges, we obtain a new hypergraph $H'$ with additive splitting functions that $(1+\epsilon)$-approximates all the cuts of $H$.
    Furthermore, $\mu_{H'} \le O\left( \epsilon^{-2} \log n \right)$.
    Hence, applying \Cref{theorem:upper-bound-sparsification-finite-spread} on $H'$, we obtain a sparsifier $H''=(V,E'')$ with at most $O(\epsilon^{-2}n\mu_{H'}\log n)=O\left( \epsilon^{-4} n\log^2 n\right)$ hyperedges.
    Note that $H''$ is a $(1+\epsilon)^2$-sparsifier of $H$ and that it has additive splitting functions.

    Finally, to bound the encoding size of $H''$ note that $K_{e''}\le n$ for all $e''\in E''$ since otherwise it does not affect the splitting function.
    Hence, we can store $K_{e''}$ using $O\left( \log n \right)$ bits.
    Additionally, the cardinality of each hyperedge $e''$ in $H''$ is at most $O\left( \epsilon^{-2}\hat{K}^{-1}\log n \right)$ by \Cref{lemma:deformation-additive}.
    Hence, we can store it using $O\left( \epsilon^{-2}\hat{K}^{-1}\log^2 n \right)$ bits (using $\log n$ bits to store every vertex in $e''$).
    Therefore, we find that the total encoding size of $H''$ is $O\left( \epsilon^{-6}\hat{K}^{-1}n\log^4 n\right)$ bits.
\end{proof}

It is straightforward to adapt the proof of \Cref{corollary:additive-small-representation} to be algorithmic.
Observe that we can assume that $|E|\le \tO(\epsilon^{-2}n^2)$, since by the structure of additive splitting functions we have $\mu_e \le n$, and thus by applying \Cref{theorem:upper-bound-sparsification-finite-spread} we can first find a $(1+\epsilon)$-sparsifier of $H$ with $\tO\left( \epsilon^{-2}n^2 \right)$ hyperedges.
We can then improve the success probability of \Cref{lemma:deformation-additive} to $1-2n^{-4}$ by changing the sampling probabilities to use $\log n$ instead of $\log |e|$.
Finally, using the union bound on all $\tO(\epsilon^{-2}n^2)$ hyperedges of $H$, we obtain that the succinct representation is successfully constructed with very high probability.

\subsection{Proof of Lemma \ref{lemma:deformation-additive}}
We will use the following claim.
\begin{claim}
    \label{claim:subset-size-concentration}
    Let $S\subseteq e$ be a subset of the vertices of the hypergraph, and let $\hat{S}$ be a random subset of $S$ where each $v\in e$ is sampled independently with probability $p = \min (\frac{c\cdot\log |e|}{K\epsilon^2},1)$ for some $K\le |e|$ and $c>0$.
    Then,
    \[
        \forall \delta \ge 0,
        \quad
        \Probability{\left| |\hat{S}|  - p|S| \right| \ge \delta p|S|} 
        \le 2|e|^{-\frac{c\delta^2 |S|}{(2+\delta)K\epsilon^2}}
        .
    \]
\end{claim}	

The claim follows immediately from the Chernoff bound (\Cref{lemma:chernoff-binomial}).

\begin{proof}[Proof of \Cref{lemma:deformation-additive}]
    Denote $\epsilon' = \epsilon/4$ and let $N= q\epsilon'^{-2}|e|^2$ for some $q\ge 1$.
    Let $e$ be some hyperedge with splitting function $g_e(S)=\min\left( \left|S\right|, K \right)$ or $g_e(S)=\min\left( \left|S\right|, \left| \bar{S} \right|, K \right)$.
    If either $K \le 100 \epsilon^{-2}\log |e|$ or $\epsilon^{-2}>|e|$, then $|e|\le O(\epsilon^{-2}(|e|/K)\log |e|)$ and we are done.

    Generate the hyperedges $\{e_i\}_{i=1}^N$ by sampling the vertices of $e$ independently with probability $p=c\epsilon'^{-2}K^{-1}\log |e|$ for some constant $100>c>0$ to be determined later.
    Note that $p<1$ since $K > 100 \epsilon^{-2}\log |e|$.
    The splitting functions of the new hyperedges are then given by $g_{e_i}(S) = \frac{1}{N}\min\left( \left|S\cap e_i \right|/p, K \right)$ or $g_{e_i}(S) = \frac{1}{N}\min\left( \left|S\cap e_i \right|/p,\left|\bar{S}\cap e_i \right|/p , K \right)$ in accordance with $g_e$.
    Observe that these functions are additive by factoring out $1/p$ from all terms in the minimum.

    We begin by showing that $e_i$ has small spread and small support.
    Observe that 
    \[
        \mu_{e_i} = \frac{\max_{S\subseteq e} g_i (S)}
        {\min_{T\subseteq e} g_i (T)} 
        \le Kp
        \le c\cdot\frac{\log |e|}{\epsilon'^2}
        .
    \]

    By Claim \ref{claim:subset-size-concentration} the cardinality of each hyperedge is no more than $2p|e|=2c\epsilon'^{-2}(|e|/K)\log |e|$ with probability $1-2|e|^{-c\delta\epsilon'^{-2} |e|/(3K)} \ge 1-2|e|^{-c\epsilon'^{-2}/3}$, where the inequality is by $K\le |e|$.
    Choosing $c\ge 20$ and applying the union bound over all $N$ hyperedges we obtain that the cardinality of all hyperedges is at most $2c\epsilon'^{-2}(|e|/K)\log |e|$ with probability at least $1-2|e|^{-c\epsilon'^{-2}/3}q\epsilon'^{-2}|e|^2 \ge 1-2|e|^{-5}$;
    the inequality holds as $\epsilon'^{-2}<|e|$, $q\ge 1$ and $\epsilon<1$ which implies $\epsilon'\le 1/4$.

    We now show that the sum $\sum_i g_{e_i}$ $(1+\epsilon)$-approximates $g_e$.
    Denote $h(S) = \sum_{i=1}^N g_{e_i}(S)$.
    We will show that for every $S\subseteq e$, with high probability,
    \[
        h(S)
        \in \left(1\pm\epsilon \right)\cdot g_e(S)
        .
    \]

    Fix some  $S\subseteq e$ such that $g_e(S)>0$ and assume without loss of generality $|S| \le |\bar{S}|$.
    The analysis of the expectation of $g_{e_i}$ is split into two claims for the monotone and symmetric cases, that we will prove shortly.
    \begin{claim}
        \label{claim:expectation-monotone-deformation}
        Let $g_{e}(S)=\min\left( \left|S\right|, K \right)$ and suppose $c>0$ is a sufficiently large constant.
        Then
        \[
            \Exp{N\cdot g_{e_i}(S)}{} \in (1\pm 2\epsilon') g_e(S)
            .
        \]
    \end{claim}

    \begin{claim}
        \label{claim:expectation-symmmetric-deformation}
        Let $g_e(S)=\min\left( \left|S\right|, \left| \bar{S} \right| , K \right)$ and suppose $c>0$ is a sufficiently large constant.
        Then 
        \[
            \Exp{N\cdot g_{e_i}(S)}{} \in (1\pm 2\epsilon') g_e(S)
            .
        \]
    \end{claim}
    By Claims \ref{claim:expectation-monotone-deformation} and \ref{claim:expectation-symmmetric-deformation},  $\Exp{h(S)}{}\in (1\pm 2\epsilon')g_e(S)$.
    Now using the Chernoff bound for bounded random variables (\Cref{lemma:chernoff-bounded}),
    \[
        \Probability{h(S) \not \in  (1\pm \epsilon')\Exp{h(S)}{}}
        \le 2\exp
            \left(
                -\frac{\epsilon'^2 \Exp{h(S)}{}}{3r}
            \right)
        ,
    \]
    where $r$ is the maximum possible contribution of a single $g_{e_i}(S)$ to the sum.
    By the definition of $g_{e_i}$ we have $r=K/N\le |e|/N=q^{-1}\epsilon'^{2}/|e|$.
    Note also that $g_e(S) \ge 1$ and hence $\Exp{h(S)}{} \ge 1-2\epsilon'$.
    Plugging these back in we find that
    \[
        \Probability{h(S) \not \in  (1\pm \epsilon')\Exp{h(S)}{}}
        \le 2\exp(-q(1-2\epsilon') |e|/3)
        .
    \]
    Choosing $q>12$, and noting $1-2\epsilon'>1/2$, gives us that with probability at least $1-2^{-2|e|}$, 
    \[
        h(S) 
        \in (1\pm \epsilon')\cdot\left((1\pm 2\epsilon')g_e(S) \right) 
        \subseteq \left(1\pm 4\epsilon'\right) g_e(S)
        .
    \]
    
    Using a union bound over all $2^{|e|}$ possible cuts of $e$ we find that $h(S)$ is a quality $1+\epsilon$ approximation for $g_e(S)$ with probability at least $1-2^{-|e|}$.
    Finally, observe that by the union bound we have that both the cardinality of all hyperedges is bounded and the quality of the approximation is $1+\epsilon$ simultaneously with probability at least $1-2|e|^{-4}$.

    We now return to prove the claim about the expectation of $h(S)$ when $g_e$ is monotone.
    The proof for the symmetric case (\Cref{claim:expectation-symmmetric-deformation}) is similar to \Cref{claim:expectation-monotone-deformation} and appears in \Cref{appendix:proof-symmetric-additive}.

    \begin{proof}[Proof of Claim \ref{claim:expectation-monotone-deformation}]
    Let $S_i = S\cap e_i$ be the intersection between $S$ and the sampled hyperedge.
    Observe that the function $\min(x,K)$ is concave, and hence by Jensen's inequality
    \[
        \Exp{\min(\left| S_i \right|/p,K)}{} 
        \le \min(\Exp{\left| S_i \right|/p}{},K)
        =\min(\left|S\right|,K)
        = g_e(S)
        .
    \]
    It remains to prove that $\Exp{Ng_{e_i}(S)}{}\ge (1-2\epsilon')g_e(S)$.
    We split the analysis into two cases, when $|S|< K/2$ and its complement.
    Starting with the case when $|S|$ is small, observe that setting $\delta=K/|S|>2$ we have by Claim \ref{claim:subset-size-concentration} that
    \[
        \Probability{\left|S_i\right| \ge pK} 
        \le 2|e|^{-\frac{c\delta^2|S|}{(2+\delta)K\epsilon'^2}}
        \le 2|e|^{-\frac{c\delta|S|}{2K}}
        =2|e|^{-\frac{c}{2}}
        ,
    \]
    where the second inequality is by $\delta/(2+\delta)>1/2$ when $\delta>2$.
    Choosing $c>5$ we find that the probability is at most $2|e|^{-5/2}$.
    Now rewriting the expectation of $f(S)$ we have
    \begin{align*}
        \Exp{Ng_{e_i}(S)}{} 
        &= \sum_{j=0}^{|S|} 
        \min\left( \frac{j}{p},K \right)
        \Probability{\left|S_i\right| = j} 
        \ge 
        \sum_{j=0}^{\lfloor pK \rfloor}
        \frac{j}{p} 
        \Probability{\left|S_i\right| = j}
        \\
        &\ge \sum_{j=0}^{|S|} \frac{j}{p}\Probability{\left|S_i\right| = j}
        - \sum_{j=\lfloor pK \rfloor +1}^{|S|} \frac{j}{p}\Probability{\left|S_i\right| = j}
        ,
    \end{align*}
    where the first inequality is by dropping all elements with $j\ge \lfloor pK \rfloor +1$ and the second is by adding and subtracting the same terms for $j\ge \lfloor pK \rfloor +1$.
    Since $j\le|S|$ and $\sum_{j=\lfloor pK \rfloor +1}^{|S|} \Probability{\left|S_i\right| = j} = \Probability{|S_i| \ge \lfloor pK \rfloor +1}$ we find
    \begin{align*}
        \Exp{f(S)}{} 
        &\ge \frac{\Exp{\left|S_i \right|}{}}{p}
        - \frac{|S|}{p} \Probability{\left|S_i\right| \ge pK}
        \ge |S|
        - \frac{|S|}{p|e|^{5/2}} 
        .
    \end{align*}
    Observe that $|S|/p = |S|K\epsilon'^2/(c\log |e|)<|e||S|$ by $K\le |e|$ and $c>1$.
    Therefore, we find
    \[
        \Exp{Ng_{e_i}(S)}{} 
        \ge |S| - \frac{|e|^2}{10|e|^{5/2}} 
        \ge |S| - \frac{1}{10|e|^{1/2}}
        \ge (1-\epsilon')g_e(S)
        ,
    \]
    where the last inequality is by $|S| \ge 1$ for all nontrivial cuts and $\epsilon^{-2}<|e|$.

    Now we turn to the case $|S|> K/2$.
    Observe that for $Ng_{e_i}(S)=\min\left( |S_i|/p, K\right)<(1-\epsilon')g_e(S)$ we must have $|S_i| < (1-\epsilon')p|S|$.
    By \Cref{claim:subset-size-concentration} this event happens with probability at most $2|e|^{-\frac{c\epsilon'^2|S|}{3K\epsilon'^2}} \le 2|e|^{-\frac{c}{6}}$.
    Therefore,
    \[
        \Exp{Ng_{e_i}(S)}{} 
        \ge \left( 1- 2|e|^{-\frac{c}{6}}\right)(1-\epsilon')p|S|/p
        \ge (1-2\epsilon')|S|
        =(1-2\epsilon)g_e(S)
        ,
    \]
    where the last inequality is by $c>9$ and $\epsilon'>1/|e|$.
    
    \end{proof}
This concludes the proof of \Cref{lemma:deformation-additive}.
\end{proof}

\section{Encoding-Size Lower Bounds}
\label{sec:encoding-size-lower-bounds}
This section shows that for additive splitting functions, $\Omega(n^2)$ bits are needed to represent sparsifiers that are reweighted-subgraphs, proving \Cref{theorem:bit-lower-bound-additive-reweighted-subgraph}.
We also prove two results on the encoding size of directed hypergraphs (see \Cref{subsec:directed-hypergraph-encoding-size}): 
(1) reweighted-subgraph sparsifiers for directed hypergraphs require encoding size of $\Omega(n^3/\epsilon)$ bits,
and (2) any data structure that stores the cuts of a directed hypergraph requires $\Omega(n^2/\epsilon)$ bits.

To show \Cref{theorem:bit-lower-bound-additive-reweighted-subgraph}, we actually prove a stronger version of the theorem for all cardinality-based splitting functions with $\mu_e <n/3$.
\begin{lemma}[Stronger Version of \Cref{theorem:bit-lower-bound-additive-reweighted-subgraph}]
    \label{lemma:bit-lower-bound-reweighted-subgraph-cardinality}
    Let $\hat{g}:[n]\to \R_+$ define a cardinality-based splitting function $g_e$ with spread $\mu_e<n/3$.
    There exists a family of hypergraphs $\mathcal{H}$ with splitting function $g_e(S)=\hat{g}(|S|)$ for all $e\in E$, such that encoding a reweighted subgraph $(1+\epsilon)$-sparsifier for an input $H\in \mathcal{H}$ requires $\Omega(n^2)$ bits.  
\end{lemma}

\subsection{Proof of Lemma \ref{lemma:bit-lower-bound-reweighted-subgraph-cardinality}}
Denote the \emph{gradient} of a function $f:[n]\to \R_+$ be defined as $\Delta_i(f) \coloneqq f(i+1)-f(i)$, we will omit the function $f$ when it is clear from the context.
The proof follows immediately from the proceeding technical claim.

\begin{claim}
    Let $\hat{g}:[n]\to \R_+$ and let $t$ be the smallest integer such that $\Delta_t<\Delta_0$.
    If $t<n/3$, then there exists a family of hypergraphs $\mathcal{H}$, with splitting function $g_e(S)=\hat{g}(|S\cap e|)$ for all $e\in E$, such that encoding a reweighted-subgraph $(1+\epsilon)$-sparsifier for an input $H\in \mathcal{H}$ requires $\Omega(n^2)$ bits.
    \label{claim:bit-lower-bound-reweighted-subgraph-undirected}
\end{claim}
Note that $\Delta_i$ satisfies the following easy property.
\begin{observation}
    \label{observation:splitting-function-non-increasing}
    Let $g_e:2^e\to\R_+$ be a cardinality based splitting function which is defined by the function $\hat{g}:[e]\to \R_+$.
    Then the series defined by $\Delta_0,\Delta_1,\ldots,\Delta_{|e|-1}$ is non-increasing.
\end{observation}
\begin{proof}
    Let $S\subseteq e$ be a set of size $i$ and $T\subset S$ be a set of size $j<i$ and $x\in e\setminus\left\{ S \right\}$.
    Then by the submodularity of $g_e$ we have
    \[
        \Delta_i = 
        g_e(S\cup \left\{ x \right\}) - g_e(S) 
        \le g_e(T\cup \left\{ x \right\}) - g_e(T)
        = \Delta_j
        .
    \]
\end{proof}

\Cref{lemma:bit-lower-bound-reweighted-subgraph-cardinality} follows immediately using \Cref{observation:splitting-function-non-increasing} and \Cref{claim:bit-lower-bound-reweighted-subgraph-undirected}.
\begin{proof}
    Assume without loss of generality that $\hat{g}(1)=1$, then $\Delta_0 =1$.
    If there exists some $k<n/3$ such that $\Delta_k \le 0$ we can apply \Cref{claim:bit-lower-bound-reweighted-subgraph-undirected} for the first $k$ with $\Delta_k < \Delta_0$.
    Since $\Delta_0 =1$ and $\Delta_i$ are non-increasing we have that $\sum_{i=0}^{\mu_e} \Delta_i \le \mu_e$ and hence the average $\Delta_i$ is at most $\mu_e/(\mu_e+1)<1$.
    Therefore, there must exist at least one $\Delta_i<1$ for some $i\in [\mu_e+1] \subseteq [n/3]$.
    Apply \Cref{claim:bit-lower-bound-reweighted-subgraph-undirected} to the first such $i$.
\end{proof}

The proof of \Cref{claim:bit-lower-bound-reweighted-subgraph-undirected} boils down to a counting argument, we create a family of hypergraphs with $g_e(S)=\hat{g}(|S|)$ as their splitting function.
The vertices of each hypergraph are partitioned into three sets $V,U,W$.
Each hyperedge is defined by a union of three parts:
(1) a random subset of the vertices of $V$,
(2) a subset of $U$  that is defined by the Hadamard code,
and (3) an unsparsifiable part over the vertices $W$.

Using the unsparsifiable part we show that any reweighted-subgraph sparsifier must contain all the hyperedges.
We then show that it is possible to exactly recover the adjacency matrix over the vertices of $V$ from any hypergraph containing the same hyperedges as $H$ (up to reweighing) using only cut queries, hence $\Omega(n^2)$ bits are required to represent it.

The recovery process is based on observing that the difference between any two cuts $S,S\cup \left\{ v \right\}$ for $v\not\in S$ is given by
\[
    \mintcut_{H}(S\cup \{v\}) - \mintcut_{H}(S)
    = \sum_{e\in E} g_e(S\cup \{v\})-g_e(S)
    = \sum_{e\in E} \indic{v\in e} \Delta_{|S\cap e|} w_e
    ,
\]
where $w_e$ is the weight of $e$.
Notice that if we find some $d$ such that $\Delta_d < \Delta_{d/2}$, and a cut where exactly one hyperedge $e^*$ has $|S\cap e^*|=d$ while all the hyperedges have $|S\cap e| \le d/2$ then we can recover whether $v\in e^*$.
We create a hypergraph where such cuts exist for every hyperedge $e\in E$ and vertex $v\in V$ using the Hadamard code.

\begin{proof}[Proof of \Cref{claim:bit-lower-bound-reweighted-subgraph-undirected}]
    Let $d=2^k$ be the smallest power of two such that $d\ge t$. 
    Assume for simplicity that $d\le n/3$.
    We start by defining the family $\mathcal{H}$ of hypergraphs.
    Let $H\in \mathcal{H}$ be a hypergraph over $n$ vertices.
    Split the vertices into three sets, $V,U,W$ with $|V| = |W| = n/6$ and $|U| = 2n/3$.
    Denote the vertices in each set by $V = \left\{ v_i \right\}_{i=1}^{n/6}$, $W = \left\{ w_i \right\}_{i=1}^{n/6}$ and $U = \left\{ u_i \right\}_{i=1}^{2n/3}$.
    Notice that since the splitting functions $g_e$ are all cardinality based, they are defined by the hyperedges.
    $H$ includes exactly $n/6$ hyperedges as described below.

    We start by describing the Hadamard code words which we will use in the proof.
    If $d=1$ then for all $i\in [n/6]$ set $p_i$ to be strings of length $2n/3$ with $1$ in the $i$-th position and $0$ elsewhere.
    Otherwise, denote the words of the Hadamard code (without the all zeros and all ones words) of length $2d$ by $h_1,\ldots,h_{2d-2}$.
    Since $2d$ could be much smaller than $2n/3$ we pad $h_i$ with zeros to get words of length $2n/3$.
    Furthermore, we wish to create a hypergraph with $n/6$ hyperedges and hence if $2d<n/6$ we create $\lceil n/(12d) \rceil$ copies of each word and denote them by $p_{i,j}$, where $p_{i,j}$ is the padded version of $h_i$ shifted by $2d\cdot j$ bits to the right.
    Observe that by the properties of the Hadamard code $p_{i,j}^2=d$ and that $p_{i,j}\cdot p_{i',j'} \in \left\{ d/2,0 \right\}$.
    For simplicity, we rename $p_{i,j}$ to $p_1,\ldots,p_{n/6}$ where $p_k = p_{\lfloor k/(2d)\rfloor,k\mod 2d}$ dropping extra words if needed.

    For each $i\in \left[ n/6 \right]$ let $e_{i}$ be a union of three sets: $P_i = \left\{ u_j \in U : p_i(j) =1 \right\}$, a random subset of $n/12$ vertices of $V$ and the singleton $\left\{ w_i \right\}$.

    Every reweighted-subgraph sparsifier $H'$ for $H$ must contain every hyperedge $e\in E$ with weight in $[1-\epsilon, 1+\epsilon]$, otherwise the cuts of the singletons, $S=\left\{ w_i \right\}$ for $i\in\left[ n/6 \right]$, would not be preserved.
    Fix some reweighted-subgraph sparsifier $H'$ of $H$, and denote the weights of the hyperedges in $H'$ by $w'_e$.

    Denote the incidence matrix of $H$ corresponding to the vertices in $V$ by $B$.
    We show that it is possible to recover $B$ from the cuts of $H'$ this means that every $H\in \mathcal{H}$ requires a unique sparsifier.
    In addition, there are $\binom{n/6}{n/12}^{n/6}$ possible choices for $B$ and hence the encoding size is $\Omega(n^2)$ bits.

    Denote the set of hyperedges containing $v_i$ by $E_i = \left\{ e\in E: v_i \in e \right\}$.
    We recover each element $B_{ij}$ by examining the difference
    \[
        \beta_{ij}
        = \mintcut_{H'} \left( P_j \cup \{ v_i \} \right) - \mintcut_{H'} \left( P_j \right)
        = \sum_{e\in E_i} g_e(P_j) - g_e(P_j\cup \{ v_i \})
        .
    \]
    Recall that by the construction of the code part of the incidence matrix, for every $k\ne j$, $|P_j \cap P_k| \in \left\{ 0,d/2 \right\}$.
    Furthermore, by the definition of $d$ we have $\Delta_{d/2} = \Delta_{0}>\Delta_d$.
    Hence,
    \[
        \beta_{ij}
        = \begin{cases}
            \beta_{ij}^1 \coloneqq \sum_{e\in E_i\setminus \left\{ e_j \right\}} \Delta_{0} w'_e 
            + \Delta_{d} w'_{e_j},
            & \text{if } B_{ij} = 1 \\
            \beta_{ij}^0 \coloneqq \sum_{e\in E_i} \Delta_{0} w'_e,
            & \text{if } B_{ij} = 0. \\
        \end{cases}
    \]
    Hence, $\beta_{ij}^1<\beta_{ij}^0$.
    In addition, observe that 
    \[
        \mintcut_{H'}(\{ v_{i} \})  
        = \sum_{e\in E_i} w'_e \hat{g}(1)
        = \Delta_0\sum_{e\in E_i} w'_e 
        .
	\]
    Therefore, if $\mintcut_{H'}(\{ v_{i} \}) = \beta_{ij}$ then $B_{ij}=0$ and otherwise $B_{ij}=1$.
\end{proof}

\subsection{Directed Hypergraph Encoding Size}
\label{subsec:directed-hypergraph-encoding-size}
In \cite{OST22}, the authors provide an $\Omega\left( n^2/\epsilon \right)$ lower bound for the number of hyperedges in a reweighted-subgraph sparsifier of directed hypergraphs.
We improve on this result in two different ways:
(1) \Cref{lemma:bit-lower-bound-directed} shows that in the reweighted-subgraph sparsifier setting, encoding directed hypergraph cuts requires $\Omega\left( n^3/\epsilon \right)$ bits.
(2) \Cref{theorem:directed-hypergraph-encoding-lower-bound} proves that any encoding of directed hypergraph cuts requires $\Omega\left( n^2/\epsilon \right)$ bits (rather than hyperedges).

We begin with \Cref{lemma:bit-lower-bound-directed}.
This result is based a similar construction to \Cref{claim:bit-lower-bound-reweighted-subgraph-undirected} without the Hadamard code part.
We define a random family $\mathcal{H}$ whose vertices are partitioned into three parts $V,U,W$.
Each hyperedge is defined by a union of two parts:
(1) a random subset of the vertices of $V$,
and (2) an unsparsifiable part on the vertices of $U,W$ that is based on  the construction of \cite{OST22}.

To recover the random part of the hypergraph we again turn to comparing different cuts of $H'$.
However, since the hypergraph is much denser we need to isolate the contribution of each hyperedge using the intersection of several cuts.

\begin{lemma}
    There exists a family of hypergraphs $\mathcal{H}$ with the directed all-or-nothing splitting function, such that for every $1/(4\epsilon)<n/3$ encoding a reweighted-subgraph $(1+\epsilon)$-sparsifier for an input $H\in \mathcal{H}$ requires $\Omega(n^3)$ bits.
    \label{lemma:bit-lower-bound-directed}
\end{lemma}

\begin{proof}
    Define a hypergraph $H$ over $n$ vertices as follows.
    Partition the vertices of $H$ into three sets of equal cardinality $V,U,W$.
    Denote $V=\left\{ v_i \right\}_{i=1}^{n/3}$ and similarly for $U,W$.
    For every $i,j\in [n/3]$ and $r\in \left[ 1,2,\ldots,\frac{1}{8\epsilon} \right]$ add a hyperedge $e_{i,i+r,j}$ with tail $\left( e_{i,i+r,j}^T \right) = \left\{ u_{i},u_{i+r \mod n/3}\right\}$ and head $\left( e_{i,i+r,j}^H \right) = \left\{ w_j \right\}$.
    Note that this is the same construction as in \cite{OST22}.
    Then augment the tail of every hyperedge with a random subset of $V$, where every vertex $v\in V$ is sampled independently with probability $1/2$.

    Observe that for any cut $S\subseteq U\cup W$ the value of the cut is independent of the random bits in the head of the hyperedges.
    Hence, following the argument in \cite{OST22} any reweighted subgraph sparsifier of $H$ must include all its hyperedges.
    Let $H'$ be some reweighted subgraph sparsifier for $H$, we will show that we can recover the random part of the tail of every hyperedge from the cuts of $H'$.
    Since $H$ has $\Theta\left( n^2/\epsilon \right)$ hyperedges, and each one encodes $n/3$ random bits, storing the any reweighted subgraph sparsifier requires $\Omega\left( n^3/\epsilon \right)$ bits.

    Denote the modified weights of the hyperedges in $H'$ by $w'(e):E\to \R_+$, we will also write $w'(F) = \sum_{e\in F}w'(e)$ for $F\subseteq E$.
    For every cut set $S$, denote the set of hyperedges $e\in E$ with $g_e(S)>0$ by $E(S)$.
    Examine the cut set $S_{i,j} \coloneqq \left\{ u_i \right\} \cup \left( W\setminus \left\{ w_j \right\} \right)$ and notice that $E\left( S_{i,j} \right) = \left\{ e_{i,i+r,j} : r\in [1/(8\epsilon)] \right\} \cup \left\{ e_{i-r,i,j} : r\in [1/(8\epsilon)] \right\}$.

    We now describe the process for determining whether $v_k$ is in $e_{i,i+x,j}^T$ for some $x\in \left[ 1,2,\ldots,\frac{1}{8\epsilon} \right]$.
    Observe that $E(S_{i,j})\cap E(S_{i+x,j}) = \left\{ e_{i,i+x,j} \right\}$, and since $e_{i,i+x,j}\in H'$, we find that $v_k\in e_{i,i+x,j}^T$ if $w'\left( E(S_{i,j})\cap E(S_{i+x,j}) \cap E\left( \left\{ v_k \right\} \right) \right) > 0$.
    To find the value of $w'\left( E(S_{i,j})\cap E(S_{i+x,j}) \cap E\left( \left\{ v_k \right\} \right)\right)$, observe that for every $S\subseteq U\cup W$ we can find $w'\left( E(S) \cap E(\left\{ v_k \right\}) \right)$ by the following method.
    
    Begin by noting that,
    \begin{align*}
        \mintcut_{H'} \left( S \cup \left\{ v_k \right\} \right)	
        - \mintcut_{H'} \left( S \right) 
        &= 	w'\left( E \left(S \cup \left\{ v_k \right\} \right) \setminus E \left(S\right)\right)
        = w'\left( E(\left\{ v_k \right\}) \setminus E \left(S\right)\right)
        \\
        &=w'\left( E(\left\{ v_k \right\})\right)
        - w'\left( E \left(S\right)\cap E(\left\{ v_k \right\}) \right)
        ,
    \end{align*}
    where the last equality is by the directed all-or-nothing splitting function and that $V$ is disjoint from the heads of the hyperedges.
    Furthermore, observe that $\mintcut_{H'}(\left\{ v_k \right\}) = w'\left( E(\left\{ v_k \right\}) \right)$ and hence, 
    \begin{equation}
        w'\left( E(S) \cap E(\left\{ v_k \right\}) \right) 
        = 	\mintcut_{H'}(\left\{ v_k \right\}) - 
        \left( \mintcut_{H'} \left( S \cup \left\{ v_k \right\} \right)	
        - \mintcut_{H'} \left( S \right)  \right)
        .
        \label{eq:recover-s-cap-v-k}
    \end{equation}
    To conclude the proof observe that 
\begin{align*}
            w'\left( E(S_{i,j})\cap E(S_{i+x,j})\cap E( \left\{ v_k \right\}) \right)
            = 
            &w'\left(E(S_{i,j})\cap E( \left\{ v_k \right\})\right) 
            + w'\left(E(S_{i+x,j})\cap E( \left\{ v_k \right\}) \right)
            \\
            - 
            &w'\left( E(S_{i,j}\cup S_{i+x,j}) \cap E( \left\{ v_k \right\}) \right)
            .
\end{align*}
    We can find all the terms on the right-hand side using \Cref{eq:recover-s-cap-v-k}, and hence we can determine for every $k$ if $v_k$ is in $e_{i,i+x,j}^T$.
    Therefore, representing any reweighted subgraph sparsifier for $H$ requires $\Omega\left( n^3/\epsilon \right)$ bits. 
\end{proof}

We now turn to proving that representing directed hypergraph cuts requires $\Omega(n^2/\epsilon)$ bits in any data structure (\Cref{theorem:directed-hypergraph-encoding-lower-bound}).
The proof constructs of hypergraphs $\mathcal{H}$ based on the construction in \cite{OST22}, with an added sampling step.
We then show that every $(1+\epsilon)$-sparsifier of some hypergraph $H\in \mathcal{H}$ does not $(1+\epsilon)$-approximate the cuts of any other member of the family.
Therefore, every hypergraph in the family requires a unique sparsifier.
The lower bound follows by showing there are $2^{\Omega(n^2/\epsilon)}$ hypergraphs in the family.
\begin{proof}[Proof of Theorem \ref{theorem:directed-hypergraph-encoding-lower-bound}]
    Define a family of hypergraphs $\mathcal{H}$ as follows.
    Let $H\in \mathcal{H}$ be a hypergraph with $2n$ vertices, and partition its vertex set into two disjoint sets of equal cardinality $U,W$.
    Throughout the proof we assume that $1/(16\epsilon)$ is an integer for simplicity.
    For every $i,j\in [n]$ sample a uniform subset of size $\frac{1}{16\epsilon}$ from $\left[ \frac{1}{8\epsilon} \right]$ and denote it by $V_{i,j}$.
    For every $x\in V_{i,j}$ add the hyperedges $e_{i,i+x,j}$ with tail $e^T_{i,i+x,j} = \left\{u_i ,u_{i+x\mod n}\right\}$ and head $e^H_{i,i+x,j}=\left\{ w_j \right\}$ to $H$.

    Fix some $H\in \mathcal{H}$.
    We will show that every quality $(1+\epsilon)$-sparsifier for $H$ does not $(1+\epsilon)$-approximate any other hypergraph $\hat{H}\in \mathcal{H}$.
    Therefore, this family implies the existence of at least $\Omega\left(2^{n^2/\epsilon} \right)$ distinct sparsifiers.
    Hence, representing a $(1+\epsilon)$-approximation of the cuts of any $H\in\mathcal{H}$ requires $\Omega\left( n^2/\epsilon \right)$ bits.

    Assume without loss of generality that the hyperedge $e_{1,2,1}$ is in $H$ but not in $\hat{H}$.
    Let $S_i = \left\{ u_i \right\} \cup \left\{ W\setminus\left\{ w_1 \right\} \right\}$.
    Observe that by the symmetry of the construction,
    \[
        \mintcut_H\left( S_i \right)	\in \left\{ \frac{1}{16\epsilon},\frac{1}{16\epsilon}+1,\ldots, \frac{1}{8\epsilon} \right\}
        .	
    \]
    If there exists some $i$ such that $\mintcut_H\left( S_i \right) \ne \mintcut_{\hat{H}}\left( S_i \right)$ then any quality $1+\epsilon$ sparsifier for $\hat{H}$ does not approximate $H$, since
    \[
        \left| \frac{\mintcut_{\hat{H}}\left( S_i \right)}
        {\mintcut_H\left( S_i \right)} -1 \right|
        \ge \left| \frac{\frac{1}{8\epsilon}-1}{\frac{1}{8\epsilon}} -1 \right|
        = 8\epsilon
        .
    \]
    Hence, we proceed with the case where $H,\hat{H}$ have the same value for every cut $S_i$.
    Assume there exists some quality $(1+\epsilon)$-sparsifier $H'$ for both $H,\hat{H}$.
    Now observe that for $H$
    \[
        \mintcut_H\left( S_1\cup S_2 \right) = 
        \mintcut_H\left( S_1 \right) + \mintcut_H\left( S_2 \right) - 1
        ,
    \]
    and since $H'$ is a quality $(1+\epsilon)$-sparsifier for $H$, then
    \[
        \mintcut_{\hat{H}}\left( S_1\cup S_2 \right) 
        \ge (1-\epsilon)
        \left( \mintcut_H\left( S_1 \right) + \mintcut_H\left( S_2 \right) - 1 \right)
        .
    \]
    We can also get an upper bound on the cut $S_1 \cup S_2$ in $\hat{H}$ by observing that since $e_{1,2,1}\not\in \hat{H}$ then
    \[
        \mintcut_{\hat{H}}\left( S_1\cup S_2 \right) = 
        \mintcut_{\hat{H}}\left( S_1 \right) + \mintcut_{\hat{H}}\left( S_2 \right)
        = 	\mintcut_{H}\left( S_1 \right) + \mintcut_{H}\left( S_2 \right)
        ,
    \]
    where the second equality is by our assumption that $H,\hat{H}$ have the same value for every cut $S_i$.
    Since $H'$ also $(1+\epsilon)$-approximates the cuts of $\hat{H}$ we have
    \[
        \mintcut_{H'}\left( S_1\cup S_2 \right)
        \le (1+\epsilon)\left( \mintcut_{H}\left( S_1 \right) + \mintcut_{H}\left( S_2 \right) \right)
        .
    \]
    However,
    \[
        (1-\epsilon)
        \left( \mintcut_H\left( S_1 \right) + \mintcut_H\left( S_2 \right) - 1 \right)
        \le (1+\epsilon)\left( \mintcut_{H}\left( S_1 \right) + \mintcut_{H}\left( S_2 \right) \right)
        ,
    \]
    whenever $\mintcut_H\left( S_1 \right) + \mintcut_H\left( S_2 \right) \le \frac{1}{4\epsilon}$
    Therefore $H'$ does not $(1+\epsilon)$ approximates both $H,\hat{H}$ and every hypergraph in $\mathcal{H}$ requires a unique sparsifier.
\end{proof}

\section{Deformation Lower Bounds}
\label{sec:deformation-lower-bounds}
In this section we prove lower bound on the support size for approximating several families of splitting functions.
In particular, we show a lower bound for additive splitting function (\Cref{theorem:lower-bound-support-additive-functions}).
The results are all based on the following technical lemma, which we prove at the end of the section.
\begin{lemma}
    \label{lemma:lower-bound-support-size}
    Let $e$ be a hyperedge with a splitting function $g_e:2^e\to\R_+$. 
    For every $S,T\subseteq e$ such that $|S|=|T|=t$, denote 
    \[
        \delta_t(S,T) \coloneqq 1-\frac{g_e(S \cup T)}{g_e(S)+g_e(T)}	
        .
    \]
    If for some $t<|e|/2$ at least a $\rho$-fraction of the pairs $(S,T)\in \binom{\binom{e}{t}}{2}$ satisfy $\delta_t(S,T)>\hat{\delta}$ for some $\hat{\delta}$ such that $\rho\hat{\delta}^2 \ge \Omega(|e|^{-1/2})$, then every $(1+\hat{\delta}/2)$-approximation of $e$ must have support size at least $\Omega\left( \rho\hat{\delta}^2 |e|/t \right)$.
\end{lemma}
Informally, the lemma states that if a splitting function is far from linear on a large enough fraction of pairs of sets of size $t$, then it cannot be closely approximated by a sum of hyperedges with small support.
The lemma is based on identifying 
\[
    \delta_t(S,T) \coloneqq 1- \frac{g_e(S\cup T)}{g_e(S)+g_e(T)}
\]
as a quantity that describes how close to linear is the function $g_e$ for subsets of size $t$, $S,T\subseteq e$.
We then show that if the function is far from linear on a large enough fraction, of pairs of sets of size $t$, then it cannot be closely approximated by a sum of hyperedges with small support.

Note that $\delta_t$ is related to the notion of curvature of submodular functions, the \emph{total curvature} of a submodular function $g_e$ is given by
\[
    c_{g_e} \coloneqq 1- \min_{S\subseteq e,v\in e\setminus S}\frac{g_e(S\cup\{ v\})-g_e(S)}{g_e(\{v\})}
    .
\]
Intuitively, the curvature describes how far from linear the function $g_e$ is in the worst case.
The curvature is used to parametrize the quality of approximation in maximization of submodular functions;
where if a function has low curvature, hence it is close to linear, then it is possible to achieve a better approximation \cite{CC84,vondrak2010submodularity}.

The quantity $\delta_t$ differs from the curvature in two regards.
First, $\delta_t$ describes a relation two sets of size $t$ and not the marginal contribution of adding a single element.
Second, in the optimization setting the guarantees depend on the worst case curvature, while our lemma requires $\delta_t$ to be large only on a constant fraction of subsets of size $t$.

\subsection{Support Size Lower Bounds for Approximating Splitting Functions}
This section proves support size lower bounds for approximating several families of common splitting functions.
A summary of the results is provided in \Cref{table:summary-of-support-size-lower-bounds}. 

We begin by presenting several results for different families of cardinality based splitting functions.
For cardinality based splitting functions the value $\delta_t(S,T)$ only depends on $|S\cup T|$, therefore it is possible to find a uniform bound on $\delta_t(S,T)$ for all sets of size $t$.
This idea is formalized in the following corollary of \Cref{lemma:lower-bound-support-size}.
\begin{corollary}[\Cref{lemma:lower-bound-support-size} for Cardinality Based Splitting Functions]
    Let $e$ be a hyperedge with a cardinality based splitting function $g_e:2^e\to\R_+$. 
    For every $t\le|e|$ denote 
    \[
        \bar{\delta}_t 
        \coloneqq 1-\max_{S_1,S_2\subseteq V: |S_1|=|S_2|=t} \frac{g_e(S_1 \cup S_2)}{g_e(S_1)+g_e(S_2)}	
        .
    \]
    Suppose $\bar{\delta}_t > 0 $ for some $t\le |e|/2$, then every $(1+\bar{\delta}_t/2)$-approximation of $e$ must have support size at least $\Omega\left( \bar{\delta}_t^2 |e|/t \right)$.
    \label{corollary:lower-bound-support-size-uniform}
\end{corollary}
\begin{proof}
    Let $\rho = 1$, and observe that $\delta_t(S,T)\ge \bar{\delta}_t$ for every $S,T\subseteq e$ of size $t$.
    Applying \Cref{lemma:lower-bound-support-size} concludes the proof.
\end{proof}

We begin with the lower bound for additive splitting functions (\Cref{theorem:lower-bound-support-additive-functions}).
\begin{proof}
    Let $e$ be a hyperedge with an additive splitting function $g_e$,  with parameter $K$.
    We will show that every $1.1$-approximation of $g_e$ requires support size $\Omega\left( |e|/K \right)$.
    Note that if $K>|e|/2$ then $\Omega(|e|/K)=\Omega(1)$ and the lower bound is trivial.
    Otherwise, note that for $t=K$ we have $\bar{\delta}_t \ge 1/2$.
    Applying \Cref{corollary:lower-bound-support-size-uniform} we find that every $1.1$-approximation of $e$ requires support size at least $\Omega\left( |e|/K \right)$. 
\end{proof}
We also provide results for polynomial and logarithmic cardinality based splitting functions.
Both results are based on identifying a constant $t$ such that $\bar{\delta}_t$ is strictly positive constant.
\begin{corollary}[Lower Bound for Polynomial Cardinality Based Splitting Functions]
    \label{corollary:lower-bound-polynomial}
    Let $\hat{g}(S) = |S|^\beta$ or $\hat{g}(S) = \min\left( |S|^\beta, \left|\bar{S}\right|^\beta \right)$ with $ \beta\in (0,0.999)$ and let $e$ be a hyperedge with $g_e(S) = \hat{g}(|S|)$ as its splitting function.
    Every $\left( 1 +(2^{-1}-2^{\beta-2}) \right)$-approximation of $g_e$ must have support size at least $\Omega\left( (2^{-1}-2^{\beta-2})^2 |e| \right)$.
\end{corollary}
\begin{proof}
    Notice that for any $t\ge 1$, we have
    \[
        \bar{\delta}_t 
        = 1- \frac{(2t)^{\beta}}{2t^\beta}
        = 1- 2^{\beta-1}
        .
    \]
    Applying \Cref{corollary:lower-bound-support-size-uniform} with $t=1$ concludes the proof.
\end{proof}
\begin{corollary}
    \label{corollary:lower-bound-polylog}
    Let $\hat{g}(S) = \log\left( |S| \right)$ or $\hat{g}(S) = \min\left( \log\left( |S| +1 \right), \log( \bar{|S|}+1 ) \right)$ and let $e$ be a hyperedge with $g_e(S) = \hat{g}(|S|)$ as its splitting function.
    Then every $1+1/5$-approximation of $g_e$ must have support size at least $\Omega\left(|e| \right)$.
\end{corollary}
\begin{proof}
    Observe that for $t=e^5$, we have
    \[
        \bar{\delta}_t 
        = 1- \frac{\log(2e^5)}{2\log(e^5)}
        >2/5
        .
    \]
    Applying \Cref{corollary:lower-bound-support-size-uniform} with $t=e^5$ concludes the proof.
\end{proof}

We also present a general lower bound for all cardinality based splitting functions that is characterized by the spread.
The proof is based on showing that if $\delta_t$ is small for all $t<r$, for $r\in \N$, then $g_e(2^r)\ge c^r$ with $c>1$.
Hence, there exists some $t<\log(\mu_e)$ with a large $\delta_t$.
We can then apply \Cref{corollary:lower-bound-support-size-uniform} for this $t$.
\begin{corollary}
    \label{corollary:lower-bound-cardinality}
    Let $e$ be a hyperedge with cardinality based splitting function $g_e(S)=\hat{g}(|S|)$, for some $\hat{g}:[|e|]\to \R_+$.
    For every $\epsilon<1/4$, every $(1+\epsilon)$-approximation of $e$ requires support size at least $\Omega\left( \epsilon^2 |e|/\mu_e^{\log_2^{-1}(2-4\epsilon)} \right)$.
\end{corollary}
\begin{proof}
    Assume without loss of generality that $\hat{g}(1)=1$.
    Note that for every $t$, if $\bar{\delta}_t \le 2\epsilon$ then
    \[
        2\epsilon
        \ge 1- \frac{\hat{g}(2t)}{2\hat{g}(t)}
        ,
    \]
    and hence $\hat{g}(2t)\ge (2-4\epsilon) \hat{g}(t)$.
    Therefore, if $\bar{\delta}_t\le 2\epsilon$ for all $t\le r$ with $r\in \N$, then $\hat{g}(2^r) \ge \left( 2-4\epsilon \right)^r$.
    However, since the maximum of the splitting function $\max_{i\in [|e|]} \hat{g}(i) = \mu_e$ there exists some $t <  2^{ \log_{2-4\epsilon}(\mu_e)}+1 \le  \mu_e^{ \log_2^{-1}(2-4\epsilon)} +1$ such that $\bar{\delta}_t > 2\epsilon$.
    The lower bound follows from applying \Cref{corollary:lower-bound-support-size-uniform} for this $t$.
\end{proof}

Finally, we present a generalization of \Cref{corollary:lower-bound-cardinality} to all \emph{unweighted splitting functions}.
A splitting function is called unweighted if all its singleton cuts are equal to $1$, i.e. $g_e(\{v\})=1$ for all $v\in e$.
One natural example of a family of unweighted splitting functions are matroid rank functions.

The proof is similar to the cardinality based case, but in the unweighted case we have an additional complication as not all cuts of size $t$ have the same value.
Therefore, we lower bound both the value of the splitting function for sets of size $t$ (as in the cardinality based case) and the fraction of sets $S$ of size $t$ for which $g_e(S)$ is at least this value.
To simplify the proof we focus on pairs of sets of size $t$ that are disjoint.
Specifically, we show that for some $t\le \mu_e^{\Omega(1)}$, at least a $\mu_e^{-\Omega(1)}$ fraction of disjoint pairs of sets have $\delta_t(S,T)>2\epsilon$.
This technique introduces an additional $\mu_e^{\Omega(1)}$ factor in the lower bound in comparison to \Cref{corollary:lower-bound-cardinality} as we apply \Cref{lemma:lower-bound-support-size} for only $\rho = \mu_e^{-\Omega(1)}$ fraction of pairs.
\begin{corollary}
    \label{corollary:lower-bound-unweighted}
    Let $\epsilon<1/4$ and denote $\gamma = \log_2^{-1}(2-4\epsilon)$.
    In addition, let $e$ by a hyperedge with an unweighted splitting function $g_e(S)$ such that $\mu_e<|e|^{1/(2\gamma)}$.
    Then, every $(1+\epsilon)$-approximation of $e$ must have support size at least $\Omega\left( \epsilon^2|e|/\mu_e^{2\gamma} \right)$.
\end{corollary}
\begin{proof}
    Let $D_i = \left\{ (S,T) \in \binom{e}{i} \times \binom{e}{i}: S\cap T = \emptyset \right\}$ be the set of disjoint pairs of subsets of size $i$ of $e$.
    Also, let $P_i = \left\{ (S,T) \in D_i:\delta_i(S,T) < 2\epsilon \right\}$ be subset of $D_i$ composed of all pairs with $\delta_i\left( S,T \right)<2\epsilon$.
    Denote $\alpha=2-4\epsilon$ and let $R_i= \left\{ (S,T)\in P_{i}: f(T) \ge \alpha^{i} ,f(S) \ge \alpha^{i} \right\}$ be the subset of $P_i$ such that $f(S),f(T)\ge\alpha^i$.
    Observe that by the definition of $\mu_e$, $R_{\log_\alpha(2\mu_e)}$ must be empty.
    
    Fix $\beta = \mu_e^{-\gamma}/16$.
    We will show that if $|P_{2^j}| \ge (1-\beta) |D_{2^j}|$ for all $j\le \log_\alpha(2\mu_e)$ then $R_{\log_\alpha(2\mu_e)}$ is nonempty and hence this leads to contradiction.
    Therefore, there exists some $q<2^{\log_\alpha(2\mu_e)}\le (2\mu_e)^{\gamma}$ such that $|P_q| \le (1-\beta) |D_q|$.
    Hence, at least $\beta|D_q|$ disjoint pairs of size $q$ have $\delta_t(S,T)>2\epsilon$.
    Furthermore, for all $q<\sqrt{|e|}$, we have that $|D_q|$ is at least $1/10$ fraction of all pairs of sets of size $t$ by the following claim, which we prove later.
    \begin{claim}
        \label{claim:many-disjoint-pairs}
        For all $q<\sqrt{|e|}$, $|D_q|> \binom{\binom{|e|}{q}}{2}/10$.
    \end{claim}
    Note that we can apply the claim as $q\le \mu_e^{\gamma}< \sqrt{|e|}$ by the theorem statement.
    Hence, at least $\beta/10$ fraction of pairs $(S,T)$ of size $q$ have $\delta_q(S,T)>2\epsilon$.
    Finally, using \Cref{lemma:lower-bound-support-size} with $t=q\le (2\mu_e)^\gamma$, $\hat{\delta} = 2\epsilon$ and $\rho = \beta/10 = \mu^{-\gamma}/160$ we find that every $(1+\epsilon)$-approximation of $e$ requires support size $\Omega\left(\epsilon^2 |e|/\mu_e^{2\gamma} \right)$.
    Note that we can apply the theorem as $\rho\hat{\delta}^2 \ge \Omega(|e|^{-1/2})$.

    To finish the proof we now show that $|P_{2^j}| \ge (1-\beta) |D_{2^j}|$ for all $j\le \log_\alpha(2\mu_e)$ implies that $R_{\log_\alpha(2\mu_e)}$ is nonempty.
    Denote the size of $|R_{2^j}| = (1-f(j))|D_{2^j}|$, we will define $f(j)$ recursively, noting that $f(0)=0$ since $g_e$ is unweighted.
    Let $S\in \binom{e}{2^{j}}$, and observe that $S$ can be partitioned into two disjoint subsets $S_1,S_2$ of size $2^{j-1}$.
    If both these subsets are in the intersection of $R_{2^{j-1}}$ and $P_{2^{j-1}}$ then 
    \begin{align*}
        f(S) = 
        f(S_1\cup S_2)
        \ge \alpha	\left( g_e(S_1)+g_e(S_2) \right)
        \ge \alpha^{r+1}
        ,
    \end{align*}
    where the first inequality is by $\delta_{2^{j-1}}(S_1,S_2)<2\epsilon$ as $(S_1,S_2)\in P_{2^{j-1}}$ and the second inequality is since for every $(S_1,S_2)\in R_{2^{j-1}}$ we have $g_e(S_1),g_e(S_2)\ge \alpha^{2^{j-1}}$.
    Hence, $S$ can be a member in pairs of $R_{2^j}$.
    By our assumption $|P_{2^{j-1}}| > (1-\beta)|D_{2^{j-1}}|$, and hence $R_{2^{j-1}}\cap P_{2^{j-1}} \ge (1-f(j-1)-\beta)|D_{2^{j-1}}|$.
    Therefore, the fraction of disjoint pairs $(S,T)\in D_{2^{j}}$ where both $S,T$ can be partitioned into disjoint subsets $S_1,S_2$ and $T_1,T_2$ that are in $R_{2^{j-1}}\cap P_{2^{j-1}}$ is at least $\left( 1-f(j-1)-\beta \right)^2$.
    We can now bound $f(j)$ by,
    \[
        f(j)
        = 2f(j-1)+2\beta-(f(j-1)+\beta)^2
        \le 2(f(j-1)+\beta)
        .
    \]
    Solving this recursive formula we find $f(j) \le 2^{j+1}\beta$.
    Recalling $\beta= \mu_e^{-\gamma}/16$ we get $f(\log_\alpha(2\mu_e))> 1/2$ and hence $R_{\log_\alpha(2\mu_e)}$ is nonempty in contradiction.
    Therefore, there exists some $q<\mu_e^\gamma$ such that $|P_q| \le (1-\beta) |D_q|$.

    To finish the proof we turn back to proving \Cref{claim:many-disjoint-pairs}.
    \begin{proof}
        Let $S,T$ be two random subsets of size $j$ of $e$.
        Denote the event that $S,T$ are disjoint by $D_{S,T}$.
        Observe that
        \begin{align*}
            \Probability{D_{S,T}} 
            &= \binom{|e|-j}{|e|}\binom{|e|}{j}^{-1}
            = \frac{(|e|-j)!(|e|-j)!}{|e|!(|e|-2j)!}
            = \prod_{i\in \left\{ |e|,|e|-1,\ldots,|e|-j+1 \right\}} \frac{i-j}{i}
            \\
            &= \prod_{i\in \left\{ |e|,|e|-1,\ldots,|e|-j+1 \right\}} 1-\frac{j}{i}
            \ge \left( 1-\frac{2j}{|e|} \right)^j
            \ge \left( 1-\frac{2}{\sqrt{|e|}} \right)^{\sqrt{|e|}}
            \ge e^{-2}\left( 1-\frac{4}{\sqrt{|e|}} \right)
            ,
        \end{align*}
        where the first inequality is by $|e|-j>n/2$ for all $j<\sqrt{|e|}$, the second by $j<\sqrt{|e|}$ and the third by $(1-x/k)^k\ge e^{-x}(1-x^2/|e|)$.
        This expression is larger than $1/10$ for all $|e|>250$.
        Therefore, the fraction of disjoint pairs out of all pairs of subsets of size $j$ is at least $1/10$.
    \end{proof}
    This concludes the proof of \Cref{corollary:lower-bound-unweighted}.
\end{proof}

\subsection{Proof of Lemma \ref{lemma:lower-bound-support-size}}
We now return to proving \Cref{lemma:lower-bound-support-size}.
\begin{proof}[Proof of \Cref{lemma:lower-bound-support-size}]
    Throughout the proof we denote $|e|=n$.
    Let $p = \alpha n/t$ for $\alpha>0$ to be determined later.
    Assume on the contrary that there exists a set of hyperedges of cardinality at most $p$ such that the sum of their splitting functions approximates $g_e$ with quality $q=1+\hat{\delta}/2$.
    Let $\left\{ e_i \right\}_{i=1}^k$ be all the possible subsets of $e$ of size $p$ and denote $k = \binom{n}{p}$.
    Note that we can assume that $p\ge 2$ as otherwise the lower bound is trivial.
    Since a sum of submodular functions is submodular, this is the most general case for decomposing $e$ into hyperedges with maximal support $p$ as any two hyperedges $e_1,e_2$ with $e_1\subseteq e_2$ can be combined into a single hyperedge with $g_{e_1}+g_{e_2}$ as its splitting function.

    Starting with some notation let $I_{T} \coloneqq \left\{ i \in [k]: e_i \cap T \ne \emptyset \right\}$ and $h(S,I) \coloneqq \sum_{i\in I} g_{e_i}(S\cap e_i)$.
    Choose some subsets of size $t$ of $e$ $S_1,S_2 \subseteq e$ with $\delta_t(S_1,S_2)\ge\hat{\delta}$.
    Notice that $h(S_1\cup S_2,I_{S_1}\setminus I_{S_2})=h(S_1,I_{S_1}\setminus I_{S_2})$, and hence we can write  
    \begin{align}
        h(S_1 \cup S_2 , [k])
        &= h(S_1, I_{S_1}\setminus I_{S_2})
        + h(S_2, I_{S_2}\setminus I_{S_1})
        + h(S_1 \cup S_2, I_{S_2} \cap I_{S_1})
        \nonumber
        \\
        &\le (1+\hat{\delta}/2) g_e(S_1\cup S_2)
        \le (1+\hat{\delta}/2)(1-\hat{\delta}) \left( g_e(S_1)  +g_e(S_2) \right)
        ,
        \label{eq:support-size-1}
    \end{align}
    where the first inequality is by our assumption that $h(S,[k])$ $(1+\hat{\delta}/2)-$approximates $g_e(S)$ and the second is since $\delta_t(S_1,S_2)\ge\hat{\delta}$.
    Observe that also 
    \[
        h(S_1, [k])
        = h(S_1, I_{S_1}\setminus I_{S_2})
        + h(S_1, I_{S_1}\cap I_{S_2})
        \ge (1-\hat{\delta}/2)g_e(S_1)
        ,
    \]
    where the inequality is again by our assumption that $h(S,[k])$ $(1+\hat{\delta}/2)-$approximates $g_e(S)$.
    Let $\beta_1 \coloneqq h(S_1, I_{S_1}\cap I_{S_2})/g_e(S_1)$, we can then write $h(S_1, I_{S_1}\setminus I_{S_2}) \ge (1-\hat{\delta}/2-\beta_1)g_e(S_1)$.
    Similarly observe that $h(S_2, I_{S_2}\setminus I_{S_1}) \ge (1-\hat{\delta}/2-\beta_2)g_e(S_2)$ and denote $\beta \coloneqq \max\left\{ \beta_1,\beta_2 \right\}$.
    Substituting back into Equation \eqref{eq:support-size-1} we find
    \[
        (1-\hat{\delta}/2-\beta)g_e(S_2)
        + (1-\hat{\delta}/2-\beta)g_e(S_1)
        \le (1+\hat{\delta}/2)(1-\hat{\delta}) \left( g_e(S_1)  +g_e(S_2) \right)
        .
    \]
    Therefore, we find $\beta =\max \left\{ \beta_1 ,\beta_2 \right\}\ge \hat{\delta}^2/2$.
    Recalling the definition of $\beta$ this implies that either $h(S_1, I_{S_1}\cap I_{S_2}) > \hat{\delta}^2 g_e(S_1)/2$ or $h(S_2, I_{S_1}\cap I_{S_2})> \hat{\delta}^2 g_e(S_2)/2$.
    By the theorem assumption there are $\rho\binom{\binom{n}{t}}{2}$ pairs of sets $A,B\subseteq e$ with $\delta_t(A,B)\ge \hat{\delta}$, and following the same argument in each pair at least one of $A,B$ satisfies this lower bound. 
    Hence, there must be at least one set $T\subseteq e$ of size $t$ for which the lower bound is satisfied in at least $\rho\binom{n}{t}^{-1} \binom{\binom{n}{t}}{2} \ge \rho\binom{n}{t}/4$ pairs.
    
    For every $i\in [k]$ denote the set $Q_i = \left\{ P \subseteq e: |P|=t, P\cap e_i \ne \emptyset, P\ne T \right\}$.
    Note that by symmetry $|Q_i|=|Q|$ for all $i\in I_T$.
    Furthermore, we can use the following lemma to bound the size of $|Q_i|$, the proof of the lemma is provided later.
    \begin{claim}
        \label{claim:lower-bound-Q-size}
        If $\alpha \le \rho\hat{\delta}^2/16$, then $|Q_i| < \frac{\rho\hat{\delta}^2}{8(1+\hat{\delta}/2)} \binom{n}{t}$ for every $i \in I_T$.
    \end{claim}
    Using this notation we can rewrite $h(T, [k])$ as
    \begin{align*}
        h(T,[k])	
        &= \sum_{i\in I_T} g_{e_i}(T)
        = \sum_{i\in I_T} \frac{1}{|Q_i|} \sum_{P\in Q_i} g_{e_i}(T)
        \\
        &         = \sum_{P\in \binom{e}{t}\setminus \left\{ T \right\}}  \sum_{i\in I_T \cap I_P} \frac{1}{|Q_i|} g_{e_i}(T)
        = \frac{1}{|Q|} \sum_{P\in \binom{e}{t} \setminus \left\{ T \right\}} h(T, I_T \cap I_P)
        ,
    \end{align*}
    where the second equality is since for every $g_{e_i}$ in the inner sum, $i\in I_T, I_P$.
    By the definition of $T$ we know that for at least $\rho\binom{n}{t}/4$ of the sets $P\in \binom{e}{t}$ satisfy $h(T, I_T \cap I_P) > \hat{\delta}^2/2 \cdot g_e(T)$.
    Hence, 
    \[
        h(T,[k])	
        \ge 
        \frac{\rho\hat{\delta}^2}{8|Q|}\binom{n}{t} g_e(T)
        > (1+\hat{\delta}/2) g_e(T)
        .
    \]
    Where the last inequality is by Claim \ref{claim:lower-bound-Q-size}.
    Therefore, $h(T,[k])$ does not $(1+\hat{\delta}/2)$-approximates $g_e(T)$ in contradiction to our assumption.
    It remains to prove Claim \ref{claim:lower-bound-Q-size}.
    \begin{proof}[Proof of Claim \ref{claim:lower-bound-Q-size}]
        Recall $Q_i = \left\{ P \subseteq V: |P|=t, P\cap e_i \ne \emptyset, P\ne T \right\}$, hence the number of elements in $Q_i$ is the total number of sets of size $t$ minus the number of sets that don't intersect $e_i$ minus 1.
        Formally, this is equal to $|Q_i| = \binom{n}{t} - \binom{n-p}{t}-1$.
        Examine,
        \begin{align*}
            \binom{n}{t}^{-1} \cdot \binom{n-p}{t}
            &= \frac{t! (n-t)!}{n!} \cdot \frac{(n-p)!}{t! (n-p -t)!}
            = \frac{(n-p)!}{n!} 
            \frac{(n-t)!}{(n-p -t)!}
            = \prod_{i=n}^{n-p+1} \frac{i-t}{i}
            \\
            &= \prod_{i=n}^{n-p+1} 1 - \frac{t}{i}
            \ge \prod_{i=n}^{n-p+1} 1 - \frac{t}{n-p}
            = \left( 1 - \frac{t}{n-p} \right)^{p-1}
            .
        \end{align*}
        Substituting $p=\alpha n/t$ and observing $\alpha =\rho\hat{\delta}^2/16 < 1/2$, we find
        \begin{align*}
            \binom{n}{t}^{-1} \cdot \binom{n-p}{t}
            &\ge \left( 1 - \frac{t}{n-\alpha n/t} \right)^{\alpha n/(2t)}
            =  \left( 1 - \frac{t}{n}\frac{1}{1-\alpha/t} \right)^{\alpha n/(2t)}
            \\
            &\ge \exp\left( -2\alpha \right)\left( 1-\left( \frac{1}{1-\alpha/t} \right)^2 \frac{t}{n} \right) ^{\alpha}
            ,
        \end{align*}
        where the first inequality is from $p\ge 2$.
        The second inequality stems from $(1-x/k)^k\ge e^x(1-x^2/k)$ and $(1-\alpha/t)^{-1} <2$ since $\alpha<1/2, t\ge 1$.
        We now split the analysis into the case where $t < n^{1/3}$ and its complement.
        When $t < n^{1/3}$, we have
        \[
            \binom{n}{t}^{-1} \cdot \binom{n-p}{t}
            \ge 
            \exp\left( -2\alpha \right)\left( 1-4 \frac{1}{n^{2/3}} \right)
            \ge 
            \exp\left( -3\alpha  \right)
            ,
        \]
        where the first inequality is from  and $t<n^{1/3}$, and the second is from $\left( 1-4 /n^{2/3} \right) \ge e^{-\alpha}$ whenever $\alpha \ge \Omega(1/\sqrt{n})$ and $n$ is large enough.
        For the case when $t\ge n^{1/3}$ observe that 
        \[
            1-\left( \frac{1}{1-\alpha/t} \right)^2 \frac{t}{n}
            > \frac{1}{3}
            ,	
        \]
        for all $n>N$ for some $N>0$.
        Hence, 
        \[
            \binom{n}{t}^{-1} \cdot \binom{n-p}{t}
            \ge
            \exp\left( -(2+\log(3))\alpha\right)
            \ge \exp\left( -4\alpha\right)
            ,
        \]
        Overall, we find that in both cases $\binom{n}{t}^{-1} \cdot \binom{n-p}{t} \ge e^{-4\alpha}$.
        Plugging this back to bound the size of $|Q_i|$ we find
        \[
            |Q_i| 
            = \binom{n}{t}- \binom{n-p}{t} -1
            < \left( 1-e^{-4\alpha} \right)\binom{n}{t} - 1
            .	
        \]
        To finish proving the claim we choose $\alpha$ such that
        \[
            \frac{\rho\hat{\delta}^2}{8(1+\hat{\delta}/2)}
            \ge 1-e^{-4\alpha}
            > 1-e^{-4\alpha} - \binom{n}{t}^{-1}
            .
        \]
        Moving sides, 
        \[
            e^{-4\alpha} 
            \ge
            1 - \frac{\hat{\delta}^2}{8(1+\hat{\delta}/2)}
            ,
        \]
        and using $\log(1-x)<-x$ for $x\in(0,1)$, we find that
        \[
            \alpha \le \frac{\rho\hat{\delta}^2}{16}    
        \]
        satisfies the requirement.
    \end{proof}
    This concludes the proof of \Cref{lemma:lower-bound-support-size}.
\end{proof}

{\small
  \bibliographystyle{alphaurl}
  \bibliography{linear-bib}

\newcommand{\etalchar}[1]{$^{#1}$}
\begin{thebibliography}{dCSHS16}

\bibitem[ACK{\etalchar{+}}16]{ACKQWZ16}
Alexandr Andoni, Jiecao Chen, Robert Krauthgamer, Bo~Qin, David~P. Woodruff, and Qin Zhang.
\newblock On sketching quadratic forms.
\newblock In {\em Innovations in Theoretical Computer Science}, ITCS'16, pages 311--319. ACM, 2016.
\newblock \href {https://doi.org/10.1145/2840728.2840753} {\path{doi:10.1145/2840728.2840753}}.

\bibitem[ADK{\etalchar{+}}16]{ADKKP16}
Ittai Abraham, David Durfee, Ioannis Koutis, Sebastian Krinninger, and Richard Peng.
\newblock On fully dynamic graph sparsifiers.
\newblock In {\em {IEEE} 57th Annual Symposium on Foundations of Computer Science, {FOCS}}, pages 335--344. {IEEE} Computer Society, 2016.

\bibitem[AGK14]{AGK14}
Alexandr Andoni, Anupam Gupta, and Robert Krauthgamer.
\newblock Towards {(1} + )-approximate flow sparsifiers.
\newblock In {\em Proceedings of the Twenty-Fifth Annual {ACM-SIAM} Symposium on Discrete Algorithms, {SODA}}, pages 279--293. {SIAM}, 2014.

\bibitem[AS20]{AssadiS20}
Sepehr Assadi and Sahil Singla.
\newblock Improved truthful mechanisms for combinatorial auctions with submodular bidders.
\newblock {\em SIGecom Exch.}, 18(1):19--27, 2020.
\newblock \href {https://doi.org/10.1145/3440959.3440964} {\path{doi:10.1145/3440959.3440964}}.

\bibitem[BK96]{benczur1996approximating}
Andr{\'{a}}s~A. Bencz{\'{u}}r and David~R. Karger.
\newblock Approximating \emph{s-t} minimum cuts in \emph{{\~{O}}}(\emph{n}\({}^{\mbox{2}}\)) time.
\newblock In {\em Proceedings of the Twenty-Eighth Annual {ACM} Symposium on the Theory of Computing}, pages 47--55. {ACM}, 1996.
\newblock \href {https://doi.org/10.1145/237814.237827} {\path{doi:10.1145/237814.237827}}.

\bibitem[BK15]{BK15}
Andr{\'{a}}s~A. Bencz{\'{u}}r and David~R. Karger.
\newblock Randomized approximation schemes for cuts and flows in capacitated graphs.
\newblock {\em {SIAM} J. Comput.}, 44(2):290--319, 2015.
\newblock \href {https://doi.org/10.1137/070705970} {\path{doi:10.1137/070705970}}.

\bibitem[BSS14]{BSS14}
Joshua~D. Batson, Daniel~A. Spielman, and Nikhil Srivastava.
\newblock Twice-ramanujan sparsifiers.
\newblock {\em {SIAM} Rev.}, 56(2):315--334, 2014.
\newblock \href {https://doi.org/10.1137/130949117} {\path{doi:10.1137/130949117}}.

\bibitem[BST19]{BST19}
Nikhil Bansal, Ola Svensson, and Luca Trevisan.
\newblock New notions and constructions of sparsification for graphs and hypergraphs.
\newblock In {\em 60th {IEEE} Annual Symposium on Foundations of Computer Science, {FOCS} 2019}, pages 910--928. {IEEE} Computer Society, 2019.
\newblock \href {https://doi.org/10.1109/FOCS.2019.00059} {\path{doi:10.1109/FOCS.2019.00059}}.

\bibitem[CC84]{CC84}
Michele Conforti and G{\'{e}}rard Cornu{\'{e}}jols.
\newblock Submodular set functions, matroids and the greedy algorithm: Tight worst-case bounds and some generalizations of the {Rado-Edmonds} theorem.
\newblock {\em Discret. Appl. Math.}, 7(3):251--274, 1984.
\newblock \href {https://doi.org/10.1016/0166-218X(84)90003-9} {\path{doi:10.1016/0166-218X(84)90003-9}}.

\bibitem[CCPS21]{CPCS21}
Ruoxu Cen, Yu~Cheng, Debmalya Panigrahi, and Kevin Sun.
\newblock Sparsification of directed graphs via cut balance.
\newblock In {\em 48th International Colloquium on Automata, Languages, and Programming, {ICALP} 2021}, volume 198 of {\em LIPIcs}, pages 45:1--45:21. Schloss Dagstuhl - Leibniz-Zentrum f{\"{u}}r Informatik, 2021.
\newblock \href {https://doi.org/10.4230/LIPIcs.ICALP.2021.45} {\path{doi:10.4230/LIPIcs.ICALP.2021.45}}.

\bibitem[CKKL12]{CKKL12}
Mahdi Cheraghchi, Adam~R. Klivans, Pravesh Kothari, and Homin~K. Lee.
\newblock Submodular functions are noise stable.
\newblock In {\em Proceedings of the Twenty-Third Annual {ACM-SIAM} Symposium on Discrete Algorithms, {SODA} 2012}, pages 1586--1592. {SIAM}, 2012.
\newblock \href {https://doi.org/10.1137/1.9781611973099.126} {\path{doi:10.1137/1.9781611973099.126}}.

\bibitem[CKN20]{CKN20}
Yu~Chen, Sanjeev Khanna, and Ansh Nagda.
\newblock Near-linear size hypergraph cut sparsifiers.
\newblock In {\em 61st {IEEE} Annual Symposium on Foundations of Computer Science, {FOCS}}, pages 61--72. {IEEE}, 2020.

\bibitem[CKST19]{CKST19}
Charles Carlson, Alexandra Kolla, Nikhil Srivastava, and Luca Trevisan.
\newblock Optimal lower bounds for sketching graph cuts.
\newblock In {\em Proceedings of the 13th Annual ACM-SIAM Symposium on Discrete Algorithms (SODA)}, pages 2565--2569, 2019.

\bibitem[dCSHS16]{SHS16}
Marcel~Kenji de~Carli~Silva, Nicholas J.~A. Harvey, and Cristiane~M. Sato.
\newblock Sparse sums of positive semidefinite matrices.
\newblock {\em {ACM} Trans. Algorithms}, 12(1):9:1--9:17, 2016.
\newblock \href {https://doi.org/10.1145/2746241} {\path{doi:10.1145/2746241}}.

\bibitem[DDS{\etalchar{+}}13]{DDSSS13}
Nikhil~R. Devanur, Shaddin Dughmi, Roy Schwartz, Ankit Sharma, and Mohit Singh.
\newblock On the approximation of submodular functions.
\newblock {\em CoRR}, abs/1304.4948, 2013.
\newblock \href {https://arxiv.org/abs/1304.4948} {\path{arXiv:1304.4948}}.

\bibitem[DS06]{DobzinskiS06}
Shahar Dobzinski and Michael Schapira.
\newblock An improved approximation algorithm for combinatorial auctions with submodular bidders.
\newblock In {\em Proceedings of the Seventeenth Annual {ACM-SIAM} Symposium on Discrete Algorithms, {SODA} 2006}, pages 1064--1073. {ACM} Press, 2006.

\bibitem[Fei09]{Feige09}
Uriel Feige.
\newblock On maximizing welfare when utility functions are subadditive.
\newblock {\em {SIAM} J. Comput.}, 39(1):122--142, 2009.
\newblock \href {https://doi.org/10.1137/070680977} {\path{doi:10.1137/070680977}}.

\bibitem[FHHP19]{FHHP19}
Wai-Shing Fung, Ramesh Hariharan, Nicholas J.~A. Harvey, and Debmalya Panigrahi.
\newblock A general framework for graph sparsification.
\newblock {\em SIAM J. Comput.}, 48(4):1196–1223, 2019.
\newblock \href {https://doi.org/10.1137/16M1091666} {\path{doi:10.1137/16M1091666}}.

\bibitem[FK14]{FK14}
Vitaly Feldman and Pravesh Kothari.
\newblock Learning coverage functions and private release of marginals.
\newblock In {\em Proceedings of The 27th Conference on Learning Theory, {COLT} 2014}, volume~35 of {\em {JMLR} Workshop and Conference Proceedings}, pages 679--702. JMLR.org, 2014.

\bibitem[FKV13]{FVK13}
Vitaly Feldman, Pravesh Kothari, and Jan Vondr{\'{a}}k.
\newblock Representation, approximation and learning of submodular functions using low-rank decision trees.
\newblock In {\em {COLT} 2013 - The 26th Annual Conference on Learning Theory}, volume~30 of {\em {JMLR} Workshop and Conference Proceedings}, pages 711--740. JMLR.org, 2013.

\bibitem[FV06]{FV06}
Uriel Feige and Jan Vondr{\'{a}}k.
\newblock Approximation algorithms for allocation problems: Improving the factor of 1 - 1/e.
\newblock In {\em 47th Annual {IEEE} Symposium on Foundations of Computer Science {(FOCS} 2006)}, pages 667--676. {IEEE} Computer Society, 2006.
\newblock \href {https://doi.org/10.1109/FOCS.2006.14} {\path{doi:10.1109/FOCS.2006.14}}.

\bibitem[FV16]{FV16}
Vitaly Feldman and Jan Vondr{\'{a}}k.
\newblock Optimal bounds on approximation of submodular and {XOS} functions by juntas.
\newblock {\em {SIAM} J. Comput.}, 45(3):1129--1170, 2016.
\newblock \href {https://doi.org/10.1137/140958207} {\path{doi:10.1137/140958207}}.

\bibitem[GHIM09]{GHIM09}
Michel~X. Goemans, Nicholas J.~A. Harvey, Satoru Iwata, and Vahab~S. Mirrokni.
\newblock Approximating submodular functions everywhere.
\newblock In {\em Proceedings of the Twentieth Annual {ACM-SIAM} Symposium on Discrete Algorithms, {SODA} 2009}, pages 535--544. {SIAM}, 2009.

\bibitem[GHRU13]{GHRU13}
Anupam Gupta, Moritz Hardt, Aaron Roth, and Jonathan~R. Ullman.
\newblock Privately releasing conjunctions and the statistical query barrier.
\newblock {\em {SIAM} J. Comput.}, 42(4):1494--1520, 2013.
\newblock \href {https://doi.org/10.1137/110857714} {\path{doi:10.1137/110857714}}.

\bibitem[GK10]{GK10}
Ryan Gomes and Andreas Krause.
\newblock Budgeted nonparametric learning from data streams.
\newblock In {\em Proceedings of the 27th International Conference on Machine Learning (ICML-10)}, pages 391--398. Omnipress, 2010.

\bibitem[HKNR98]{hagerup1998characterizing}
Torben Hagerup, Jyrki Katajainen, Naomi Nishimura, and Prabhakar Ragde.
\newblock Characterizing multiterminal flow networks and computing flows in networks of small treewidth.
\newblock {\em J. Comput. Syst. Sci.}, 57(3):366--375, 1998.
\newblock \href {https://doi.org/10.1006/jcss.1998.1592} {\path{doi:10.1006/jcss.1998.1592}}.

\bibitem[JLLS23]{JLLS23}
Arun Jambulapati, James~R. Lee, Yang~P. Liu, and Aaron Sidford.
\newblock Sparsifying sums of norms.
\newblock In {\em 64th {IEEE} Annual Symposium on Foundations of Computer Science, {FOCS} 2023}, pages 1953--1962. {IEEE}, 2023.
\newblock \href {https://doi.org/10.1109/FOCS57990.2023.00119} {\path{doi:10.1109/FOCS57990.2023.00119}}.

\bibitem[JRT24]{JRT24}
Arun Jambulapati, Victor Reis, and Kevin Tian.
\newblock Linear-sized sparsifiers via near-linear time discrepancy theory.
\newblock In {\em Proceedings of the 2024 Annual ACM-SIAM Symposium on Discrete Algorithms (SODA)}, pages 5169--5208. SIAM, 2024.

\bibitem[Kar93]{Karger93}
David~R. Karger.
\newblock Global min-cuts in rnc, and other ramifications of a simple min-cut algorithm.
\newblock In {\em Proceedings of the Fourth Annual {ACM/SIGACT-SIAM} Symposium on Discrete Algorithms}, pages 21--30. {ACM/SIAM}, 1993.

\bibitem[KG11]{KG11}
Andreas Krause and Carlos Guestrin.
\newblock Submodularity and its applications in optimized information gathering.
\newblock {\em {ACM} Trans. Intell. Syst. Technol.}, 2(4):32:1--32:20, 2011.
\newblock \href {https://doi.org/10.1145/1989734.1989736} {\path{doi:10.1145/1989734.1989736}}.

\bibitem[KK15]{KK15}
Dmitry Kogan and Robert Krauthgamer.
\newblock Sketching cuts in graphs and hypergraphs.
\newblock In {\em Proceedings of the 2015 Conference on Innovations in Theoretical Computer Science, {ITCS} 2015}, pages 367--376. {ACM}, 2015.
\newblock \href {https://doi.org/10.1145/2688073.2688093} {\path{doi:10.1145/2688073.2688093}}.

\bibitem[KKTY21]{KKTY21}
Michael Kapralov, Robert Krauthgamer, Jakab Tardos, and Yuichi Yoshida.
\newblock Towards tight bounds for spectral sparsification of hypergraphs.
\newblock In {\em {STOC} '21: 53rd Annual {ACM} {SIGACT} Symposium on Theory of Computing}, pages 598--611. {ACM}, 2021.
\newblock \href {https://doi.org/10.1145/3406325.3451061} {\path{doi:10.1145/3406325.3451061}}.

\bibitem[KPS24]{KPS24}
Sanjeev Khanna, Aaron Putterman, and Madhu Sudan.
\newblock Code sparsification and its applications.
\newblock In {\em Proceedings of the 2024 Annual ACM-SIAM Symposium on Discrete Algorithms (SODA)}, pages 5145--5168. SIAM, 2024.

\bibitem[KPZ19]{karpov2019exponential}
Nikolai Karpov, Marcin Pilipczuk, and Anna Zych{-}Pawlewicz.
\newblock An exponential lower bound for cut sparsifiers in planar graphs.
\newblock {\em Algorithmica}, 81(10):4029--4042, 2019.
\newblock \href {https://doi.org/10.1007/s00453-018-0504-8} {\path{doi:10.1007/s00453-018-0504-8}}.

\bibitem[KR13]{krauthgamer2013mimicking}
Robert Krauthgamer and Inbal Rika.
\newblock Mimicking networks and succinct representations of terminal cuts.
\newblock In {\em Proceedings of the Twenty-Fourth Annual ACM-SIAM Symposium on Discrete Algorithms}, SODA '13, page 1789–1799. SIAM, 2013.

\bibitem[KZ23]{KZ23}
Jannik Kudla and Stanislav Zivn{\'{y}}.
\newblock Sparsification of monotone $k$-submodular functions of low curvature.
\newblock {\em CoRR}, abs/2302.03143, 2023.
\newblock \href {https://arxiv.org/abs/2302.03143} {\path{arXiv:2302.03143}}.

\bibitem[LB11]{LB11}
Hui Lin and Jeff~A. Bilmes.
\newblock A class of submodular functions for document summarization.
\newblock In {\em The 49th Annual Meeting of the Association for Computational Linguistics: Human Language Technologies, Proceedings of the Conference 2011}, pages 510--520. The Association for Computer Linguistics, 2011.

\bibitem[LM17]{LM17}
Pan Li and Olgica Milenkovic.
\newblock Inhomogeneous hypergraph clustering with applications.
\newblock In {\em Advances in Neural Information Processing Systems 30: Annual Conference on Neural Information Processing Systems 2017}, pages 2308--2318, 2017.

\bibitem[LM18]{LM18}
Pan Li and Olgica Milenkovic.
\newblock Submodular hypergraphs: p-laplacians, cheeger inequalities and spectral clustering.
\newblock In {\em Proceedings of the 35th International Conference on Machine Learning, {ICML} 2018}, volume~80 of {\em Proceedings of Machine Learning Research}, pages 3020--3029. {PMLR}, 2018.

\bibitem[LVS{\etalchar{+}}21]{LVSLG21}
Meng Liu, Nate Veldt, Haoyu Song, Pan Li, and David~F. Gleich.
\newblock Strongly local hypergraph diffusions for clustering and semi-supervised learning.
\newblock In {\em {WWW} '21: The Web Conference 2021}, pages 2092--2103. {ACM} / {IW3C2}, 2021.
\newblock \href {https://doi.org/10.1145/3442381.3449887} {\path{doi:10.1145/3442381.3449887}}.

\bibitem[McC05]{mccormick2005submodular}
S~Thomas McCormick.
\newblock Submodular function minimization.
\newblock {\em Handbooks in operations research and management science}, 12:321--391, 2005.

\bibitem[OST23]{OST22}
Kazusato Oko, Shinsaku Sakaue, and Shin{-}ichi Tanigawa.
\newblock Nearly tight spectral sparsification of directed hypergraphs.
\newblock In {\em 50th International Colloquium on Automata, Languages, and Programming, {ICALP} 2023}, volume 261 of {\em LIPIcs}, pages 94:1--94:19. Schloss Dagstuhl - Leibniz-Zentrum f{\"{u}}r Informatik, 2023.
\newblock \href {https://doi.org/10.4230/LIPIcs.ICALP.2023.94} {\path{doi:10.4230/LIPIcs.ICALP.2023.94}}.

\bibitem[Pog17]{Pogrow17}
Yosef Pogrow.
\newblock Solving symmetric diagonally dominant linear systems in sublinear time (and some observations on graph sparsification).
\newblock Master's thesis, Weizmann Institute of Science, 2017.
\newblock URL: \url{https://www.wisdom.weizmann.ac.il/~robi/files/YosefPogrow-MScThesis-2017_12.pdf}.

\bibitem[Qua24]{quanrud2022quotient}
Kent Quanrud.
\newblock Quotient sparsification for submodular functions.
\newblock In {\em Proceedings of the 2024 Annual ACM-SIAM Symposium on Discrete Algorithms (SODA)}, pages 5209--5248. SIAM, 2024.

\bibitem[RY22]{RY22}
Akbar Rafiey and Yuichi Yoshida.
\newblock Sparsification of decomposable submodular functions.
\newblock In {\em Thirty-Sixth {AAAI} Conference on Artificial Intelligence}, pages 10336--10344. {AAAI} Press, 2022.
\newblock \href {https://doi.org/10.1609/aaai.v36i9.21275} {\path{doi:10.1609/aaai.v36i9.21275}}.

\bibitem[SS11]{SS11}
Daniel~A. Spielman and Nikhil Srivastava.
\newblock Graph sparsification by effective resistances.
\newblock {\em {SIAM} J. Comput.}, 40(6):1913--1926, 2011.
\newblock \href {https://doi.org/10.1137/080734029} {\path{doi:10.1137/080734029}}.

\bibitem[ST11]{ST11}
Daniel~A. Spielman and Shang{-}Hua Teng.
\newblock Spectral sparsification of graphs.
\newblock {\em {SIAM} J. Comput.}, 40(4):981--1025, 2011.
\newblock \href {https://doi.org/10.1137/08074489X} {\path{doi:10.1137/08074489X}}.

\bibitem[SY19]{SY19}
Tasuku Soma and Yuichi Yoshida.
\newblock Spectral sparsification of hypergraphs.
\newblock In {\em Proceedings of the Thirtieth Annual {ACM-SIAM} Symposium on Discrete Algorithms, {SODA} 2019}, pages 2570--2581. {SIAM}, 2019.
\newblock \href {https://doi.org/10.1137/1.9781611975482.159} {\path{doi:10.1137/1.9781611975482.159}}.

\bibitem[TIWB14]{TIWB14}
Sebastian Tschiatschek, Rishabh~K. Iyer, Haochen Wei, and Jeff~A. Bilmes.
\newblock Learning mixtures of submodular functions for image collection summarization.
\newblock In {\em Advances in Neural Information Processing Systems 27 (NeurIPS 2014)}, pages 1413--1421, 2014.

\bibitem[VBK20]{VBK20}
Nate Veldt, Austin~R. Benson, and Jon~M. Kleinberg.
\newblock Minimizing localized ratio cut objectives in hypergraphs.
\newblock In {\em {KDD} '20: The 26th {ACM} {SIGKDD} Conference on Knowledge Discovery and Data Mining}, pages 1708--1718. {ACM}, 2020.
\newblock \href {https://doi.org/10.1145/3394486.3403222} {\path{doi:10.1145/3394486.3403222}}.

\bibitem[VBK21]{VBK21}
Nate Veldt, Austin~R. Benson, and Jon~M. Kleinberg.
\newblock Approximate decomposable submodular function minimization for cardinality-based components.
\newblock In {\em Advances in Neural Information Processing Systems 34 (NeurIPS 2021)}, pages 3744--3756, 2021.
\newblock URL: \url{https://proceedings.neurips.cc/paper/2021/hash/1e8a19426224ca89e83cef47f1e7f53b-Abstract.html}.

\bibitem[VBK22]{VBK22}
Nate Veldt, Austin~R. Benson, and Jon~M. Kleinberg.
\newblock Hypergraph cuts with general splitting functions.
\newblock {\em {SIAM} Rev.}, 64(3):650--685, 2022.
\newblock \href {https://doi.org/10.1137/20m1321048} {\path{doi:10.1137/20m1321048}}.

\bibitem[Von10]{vondrak2010submodularity}
Jan Vondr{\'a}k.
\newblock Submodularity and curvature: The optimal algorithm (combinatorial optimization and discrete algorithms).
\newblock {\em RIMS Kokyuroku Bessatsu}, 23:253--266, 2010.
\newblock URL: \url{http://hdl.handle.net/2433/177046}.

\bibitem[Yam16]{Yamaguchi2016}
Yutaro Yamaguchi.
\newblock Realizing symmetric set functions as hypergraph cut capacity.
\newblock {\em Discret. Math.}, 339(8):2007--2017, 2016.
\newblock \href {https://doi.org/10.1016/j.disc.2016.02.010} {\path{doi:10.1016/j.disc.2016.02.010}}.

\bibitem[ZLS22]{ZLS22}
Yu~Zhu, Boning Li, and Santiago Segarra.
\newblock Hypergraph 1-spectral clustering with general submodular weights.
\newblock In {\em 56th Asilomar Conference on Signals, Systems, and Computers, {ACSSC} 2022}, pages 935--939. {IEEE}, 2022.
\newblock \href {https://doi.org/10.1109/IEEECONF56349.2022.10052065} {\path{doi:10.1109/IEEECONF56349.2022.10052065}}.

\end{thebibliography}
} %
\appendix

\section{Chernoff Bounds}
\label{sec:preliminaries}
We use the following versions of the Chernoff bound throughout the paper.
	\begin{lemma}[Chernoff bound for bernoulli random variables]
		Let $X_1,\ldots, X_m$ be independent random variables taking values in $\left\{0,1\right\}$. Let $X$ denote their sum and $\mu = \Exp{X}{}$. 
		Then,
		\[
			\forall \delta \ge 0,
			\quad
			\Probability{\left|X-\mu\right| \ge \delta \mu}
			 \le 2\cdot\exp\left(-\frac{\delta^2 \mu}{2+\delta}\right)
			 .
		\]
		\label{lemma:chernoff-binomial}
	\end{lemma}
	\begin{lemma}[Chernoff bound for bounded random variables, Theorem 6.1 in \cite{AGK14}]
		Let $X_1,\ldots, X_m\ge 0$ be independent random variables such that either $X_i$ is deterministic or $X_i\in [0,b]$.
		Let $X$ denote their sum and $\mu = \Exp{X}{}$, then,
		\[
			\forall \delta >0,
			\quad
			\Probability{\left|X-\mu\right| \ge \delta \mu}
			\le 2\cdot\exp\left(-\frac{\delta^2 \mu}{(2+\delta)b}\right)
			.
		\]
		Additionally,
		\[
			\forall \delta \in [0,1],
			\quad
			\Probability{\left|X-\mu\right| \ge \delta \mu}
			\le 2\cdot\exp\left(-\frac{\delta^2 \mu}{3b}\right)
			.
		\]
		\label{lemma:chernoff-bounded}
	\end{lemma}

\section{Approximating Coverage Functions}
\label{appendix:coverage}

Due to the wide prevalence of submodular hypergraph cut functions, our results have several applications.
One example is finding a succinct representation for coverage functions.
\begin{definition}
  A function $f:2^V\to \R_+$ is called a \emph{coverage function}
  over ground set $W$ with weight function $\rho:W\to \R_+$
  if there exists a collection $\{A_v\}_{v\in V}$ where each $A_v\subseteq W$,
  such that
  \[
    \forall S\subseteq V,
    \qquad
    f(S) = \sum_{w\in \bigcup_{v\in S} A_v} \rho(w). 
  \]
\end{definition}
Coverage functions are commonly used as objective functions,
for example in sensor-placement problems~\cite{KG11}.
The running time of algorithms for these problems
may be improved considerably by decreasing the size of the ground set $|W|$. 
To this end, the next corollary shows
that every coverage function can be approximated
by a coverage function with ground-set size $\tO_\epsilon(n)$. 
In general, the given ground set might have size $2^n$,
hence the decrease in size may be exponential. 

\begin{corollary}
  \label{claim:coverage-approximation}
  Let $e$ be a hyperedge whose splitting function $g_e$ is a coverage function,
  and let $K \eqdef \max_{w\in W} \left| \left\{v\in V: w\in A_v  \right\} \right|$
  The $e$ can be $(1+\epsilon)$-approximated by $O(\epsilon^{-2} n \log n)$ hyperedges with support size at most $K$.
  Furthermore, the resulting sparsifier is a coverage function on $O(\epsilon^{-2} n \log n)$ elements.
\end{corollary}

This result was recently obtained independently in~\cite{quanrud2022quotient},
using two different proof methods.
One of them is by reduction to (sparsification of) undirected hypergraph cuts.
Our proof is simpler, and designs a reduction to (sparsification of) additive splitting,
for which we can apply \Cref{theorem:upper-bound-sparsification-finite-spread}.

\begin{proof}
  For $w\in W$, 
  let $V_w \coloneqq \left\{v\in V: w\in A_v  \right\}$.
  Observe that $f(S)$ can be written as 
  \[
    f(S) 
    = \sum_{w\in \bigcup_{v\in S} A_v} \rho(w)
    = \sum_{w\in W} \rho(w)\cdot 1_{\left| V_w \cap S \right| > 0}
    = \sum_{w\in W} \rho(w) \cdot \min\left(\left| V_w \cap S \right|, 1 \right)
    .
  \]
  Hence, $f$ can be written as a sum of $|W|$ splitting functions of the form
  $g_w: S\mapsto \rho(w)\cdot\min\left(\left| V_w \cap S \right|, 1 \right)$.
  Observe that the spread of each $g_w$ is $\mu_{g_w} = 1$,
  therefore \Cref{theorem:upper-bound-sparsification-finite-spread} yields the desired result.
  Finally, note that the resulting sparsifier is a reweighted subgraph,
  hence it is a coverage function with $O(\epsilon^{-2} n \log n )$ elements.
\end{proof}

\section{Application to Terminal Cuts in a Graph}
\label{appendix:terminal-cut-functions}
\textbf{Terminal-Cut Functions.}
Let $G=(V,E,w)$ be some undirected graph and let $\T\subseteq V$ be a special set of vertices called the terminals of $G$.
Denote $|\T|=k$.
The terminal cut function of $G$ is defined as
\[
    \forall S\subseteq \T,
    \quad
    \mintcut_G(S,\bar{S})
    =
    \min_{U\subseteq V: U\cap \T = S} \sum_{e\in \delta(U)} w(e)
    ,
\]
where $\delta(U)$ is the set of edges with exactly one endpoint in $U$.
Note that the terminal cut function $\mintcut_G(S):2^{\T}\to \R_+$, of $G$ is submodular.
In \cite{hagerup1998characterizing}, the authors show that it is possible to construct a graph with $O(2^{2^k})$ vertices that preserves the terminal cut function of $G$ exactly.
On the other hand, \cite{krauthgamer2013mimicking,karpov2019exponential} showed that a minimum of $2^{\Omega(k)}$ vertices is necessary to preserve the terminal cut function of $G$, even for planar graphs.

In the approximate case, a construction of quality $1+\epsilon$ with $\tilde{O}(\text{poly}(k,\epsilon^{-1}))$ vertices was shown for bipartite graphs by \cite{AGK14,ADKKP16}.
However, there is no known upper or lower bound for the size of a data structure approximating the cuts of general graphs.
One such possible data structure would be to represent the terminal cut function as cuts of a submodular hypergraph with simple splitting functions.
Then, using Theorem \ref{theorem:upper-bound-sparsification-all} we can achieve a small data structure to represent the graph cuts.
However, it turns out that this is not possible using the all-or-nothing splitting function.

\begin{theorem}[Theorem 3.3 in \cite{Yamaguchi2016}]
    If a symmetric submodular function $f:2^V\to \R_+$ can be realized as a cut capacity function of an undirected hypergraph with nonnegative capacities, then
    \[	
        \forall i\in [V],
        \forall S \in \binom{V}{i},
        \quad
        (-1)^{i}f^{(i)} (S) \leq 0
        ,
    \]
    where
    \[
        f^{(i)}(S) = \sum_{X\subseteq S} (-1)^{|S \setminus X|} f(X)
        .
    \]
    \label{theorem:submodular-realization}
\end{theorem}

Consider the following counter-example, let $G=(V,E)$ be the star graph with 4 leaves and a central vertex, where the terminals $\T = \left\{ t_1,t_2,t_3,t_4 \right\}$ are the leaves.
Denoting its terminal cut function as $f_G$, observe that $f_G^{(3)}(\left\{t_1,t_2,t_3  \right\}) = -2 <0$.
Hence, by Theorem \ref{theorem:submodular-realization}, $f$ cannot be realized as a cut capacity function of an undirected hypergraph.
However, in this case $f_G$ is exactly the small-side splitting function.
Therefore, we pose the following question - is there a class of simple splitting functions that can represent the terminal cut function of any graph?

\section{Proof of Expectation for Symmetric Additive Decomposition}
\label{appendix:proof-symmetric-additive}
\begin{proof}[Proof of Claim \ref{claim:expectation-symmmetric-deformation}]
    Let $S_i = S\cap e_i, \bar{S}_i = \bar{S} \cap e_i$ be the intersection between $S$ and the sampled hyperedge.
    Assume without loss of generality that $|S|<|\bar{S}|$.
    Observe that the function $\min(x,y,K)$ is concave, and hence by Jensen's inequality
    \[
        \Exp{\min(\left| S_i \right|/p, \left| \bar{S}_i \right|/p, K)}{} 
        \le \min(\Exp{\left| S_i \right|/p}{},\Exp{\left| \bar{S}_i\right|/p}{}, K)
        =\min(\left|S\right|,\left|\bar{S}\right| K)
        = g_e(S)
        ,
    \]
    where the second inequality is since $S_i,\bar{S}_i$ are independent.
    Hence, it only remains to prove that $\Exp{Ng_{e_i}(S)}{}\ge (1-2\epsilon')g_e(S)$.
    We split the analysis into two cases, when $|S|< K/2$ and its complement.
    Starting with the case when $|S|$ is small, observe that setting $\delta=K/|S|>2$ we have by Claim \ref{claim:subset-size-concentration} that
    \[
        \Probability{\left|S_i\right| \ge pK} 
        \le 2|e|^{-\frac{c\delta^2|S|}{(2+\delta)K\epsilon'^2}}
        \le 2|e|^{-\frac{c\delta|S|}{2K}}
        =2|e|^{-\frac{c}{2}}
        ,
    \]
    where the second inequality is by $\delta/(2+\delta)>1/2$ when $\delta>2$.
    Choosing $c>10$ we find that the probability is at most $|e|^{-5}$.
    Note that $K<|e|/2$, otherwise it doesn't affect the splitting function.
    Hence, $|\bar{S}|\ge 3K/2$ therefore using \Cref{claim:subset-size-concentration} again we have (setting $\delta=1/2$),
    \[
        \Probability{\left|\bar{S}_i\right|<pK} 
        \le 2|e|^{-\frac{c\delta^2|\bar{S}|}{3K\epsilon'^2}}
        \le 2|e|^{-\frac{3cK/2}{12K\epsilon'^2}}
        \le 2|e|^{-4c}
        ,
    \]
    where the last inequality is by $\epsilon'=\epsilon/4$ and $\epsilon<1$.
    Choosing $c>2$ we find that the probability that both events not occur is at most $4|e|^{-5}$.
    Note that the expectation of $Ng_{e_i}(S)$ is given by
    \begin{align*}
        \Exp{Ng_{e_i}(S)}{}
        =
        \sum_{j=0}^{\lfloor pK \rfloor} \frac{j}{p} (\Probability{|S_i| = j, |\bar{S}_i|\ge j} + \Probability{|S_i| > j, |\bar{S}_i|=j})  
        +
        K \Probability{|S_i| > pK, |\bar{S}|>pK}
        .
    \end{align*}
    Hence, we can bound the expectation from below by keeping only the most significant terms
    \begin{align*}
        \Exp{Ng_{e_i}(S)}{}
        &\ge
        \sum_{j=0}^{\lfloor pK \rfloor} \frac{j}{p} \Probability{|S_i| = j, |\bar{S}_i|\ge j}
        \ge 
        \sum_{j=0}^{\lfloor pK \rfloor}\frac{j}{p} \Probability{|S_i| = j} \Probability{ |\bar{S}_i|\ge j}
        \\
        &\ge
        \Probability{ |\bar{S}_i|\ge pK} 
        \sum_{j=0}^{\lfloor pK \rfloor}\frac{j}{p}\Probability{|S_i| = j}
        \ge 
        \left( 1-4|e|^{-5} \right)
        \sum_{j=0}^{\lfloor pK \rfloor}\frac{j}{p} \Probability{|S_i| = j}
        ,
    \end{align*}
    where the second inequality is by the independence of $S_i,\bar{S}_i$ and the last by substituting the bound $\Probability{|\bar{S}_i|<pK}$.
    Now adding and subtracting the rest of the possible values of $|S_i|$ we get
    \begin{align*}
        \Exp{Ng_{e_i}(S)}{} 
        &\ge 
        \left( 1-4|e|^{-5} \right)
        \left( \sum_{j=0}^{|S|}\frac{j}{p} \Probability{|S_i| = j}
        -\sum_{j=\lfloor pK \rfloor +1}^{|S|}\frac{j}{p} \Probability{|S_i| = j} \right)
        \\
        &\ge
        \left( 1-4|e|^{-5} \right)
        \left( \frac{\Exp{|S_i|}{}}{p}
        -\frac{|S|^2}{p} \Probability{|S_i| >pK} \right)
        ,
    \end{align*}
    where the last inequality is by $j,|S|-\lfloor pK \rfloor <|S|$.
    Note that $\Exp{|S_i|}{}/p=|S|=g_e(S)$ since $|S|<K<|\bar{S}|$.
    Recall that $p=c\epsilon'^{2}K^{-1}\log |e|$ and observe, 
    \[
        \frac{|S|^2}{p} \Probability{|S_i| >pK} 
        \le \frac{|e|^2 K}{c|e|^5 \log |e|} 
        \le  \frac{1}{10|e|^2}
        ,
    \]
    where the first inequality is by substituting the bound for $\Probability{|S_i| >pK}$ and the second by $K<|e|$, $\log |e| >1$ and $c>10$.
    Hence, we obtain $\Exp{Ng_{e_i}(S)}{}\ge (1-4|e|^{-5})(g_e(S)-|e|^{-2}) \ge (1-2\epsilon')g_e(S)$ with the last inequality by $\epsilon^{-2}<|e|$.

    Now we turn to the case  $|S|> K/2$.
    Observe that for $Ng_{e_i}(S)=\min(|S_i|/p,|\bar{S}_i|/p,K)<(1-\epsilon')g_e(S)$ we must have $\min\left( |S_i|,\left| \bar{S}_i\right| \right) < (1-\epsilon')p|S|$.
    By \Cref{claim:subset-size-concentration} this event happens with probability at most $4|e|^{-\frac{c\epsilon'^2|S|}{3K\epsilon'^2}} \le 4|e|^{-\frac{c}{6}}$.
    Therefore,
    \[
        \Exp{Ng_{e_i}(S)}{} 
        \ge \left( 1- 4|e|^{-\frac{c}{6}}\right)(1-\epsilon')p|S|/p
        \ge (1-2\epsilon')|S|
        = (1-2\epsilon')g_e(S)
        ,
    \]
    where the last inequality is by setting $c>15$ and $\epsilon^2>1/|e|$.
    \end{proof}

\end{document}